\documentclass[lettersize,onecolumn]{IEEEtran}
\usepackage{amsmath,amsfonts,amssymb,amsthm}
\usepackage{mathrsfs}
\usepackage{xspace}
\usepackage{algorithmic}
\usepackage{algorithm}
\usepackage{mathtools}
\usepackage{array}
\usepackage{tikz}
\usepackage{caption} 
\usetikzlibrary{shapes,arrows}
\usepackage[caption=false,font=normalsize,labelfont=sf,textfont=sf]{subfig}
\usepackage{textcomp}
\usepackage{stfloats}
\usepackage{url}
\usepackage{verbatim}
\usepackage{graphicx}
\usepackage{cite}
\newcommand{\Ebb}{\mathbb{E}}

\newcommand{\Pbb}{\mathbb{P}}

\newcommand{\Acal}{\mathcal{A}}

\newcommand{\Xcal}{\mathcal{X}}
\newcommand{\Wcal}{\mathcal{W}}

\def\ie{\textit{i.e.}\@\xspace}

\newtheorem{theorem}{Theorem}[section]

\newtheorem{lemma}[theorem]{Lemma}
\newtheorem{corollary}[theorem]{Corollary}

\theoremstyle{definition}
\newtheorem{definition}[theorem]{Definition}

\usepackage{changepage}
\usepackage{tikz}
\usetikzlibrary{calc}
\usepackage{color,hyperref,listings}
\usepackage[shortlabels]{enumitem}
\usepackage{algorithm}
\usepackage{algorithmic}
\usepackage{amssymb}
\usepackage{tkz-euclide}
\usepackage{siunitx}
\usepackage{physics}
\usepackage{multicol}
\usepackage{pgfplots}
\usepackage{adjustbox}
\usepackage{cancel}
\usepackage{commath}
\usepackage{graphicx}
\usepackage{mathdots}
\usepackage{tabularx}
\usepackage{mathtools}
\usepackage{booktabs}
\usepackage{environ}
\usepackage{cleveref}
\usepackage{fancyhdr}
\usepackage{centernot}
\usepackage[margin=0.75in]{geometry}
\usepackage[american,siunitx,arrowmos]{circuitikz}
\usepackage[final]{pdfpages}
\pgfplotsset{compat=1.18}
\usepackage[english]{babel}
\usepackage{fontawesome}
\usepackage{listings}
\usepackage{fixmath}
\usepackage{xspace}
\usepackage{listings}
\usepackage{amsthm}
\usepackage{upgreek}
\usetikzlibrary{trees}
\usepackage{amsmath,amsfonts}
\usepackage{array}
\usepackage[caption=false,font=normalsize,labelfont=sf,textfont=sf]{subfig}
\usepackage{textcomp}
\usepackage{stfloats}
\usepackage{url}
\usepackage{verbatim}
\usepackage{cite}
\usepackage{graphicx}
\usepackage{prettyref}
\usepackage{wrapfig}

\begin{document}

\title{Universal Discrete Filtering with Lookahead or Delay }

\author{Pumiao Yan, Jiwon Jeong, Naomi Sagan and Tsachy Weissman

\thanks{Pumiao Yan and Jiwon Jeong are co-first authors.}
\thanks{Pumiao Yan, Jiwon Jeong, Naomi Sagan and Tsachy Weissman are with the Department of Electrical Engineering, Stanford University, CA 94305 USA (email: pumiaoy@stanford.edu; jeongjw@stanford.edu; nsagan@stanford.edu; tsachy@stanford.edu)}

}



\markboth{IEEE Transactions on Information Theory, Submitted for Review}%
{Shell \MakeLowercase{\textit{et al.}}: A Sample Article Using IEEEtran.cls for IEEE Journals}


\maketitle

\begin{abstract}
We consider the universal discrete filtering problem, where an input sequence generated by an unknown source passes through a discrete memoryless channel, and the goal is to estimate its components based on the output sequence, with limited lookahead or delay. 
We propose and establish the universality of a family of schemes for this setting.
These schemes are induced by universal Sequential Probability Assignments (SPAs), and inherit their computational properties. We show that the schemes induced by LZ78 are practically implementable and well-suited for scenarios with limited computational resources and latency constraints. 
In passing, we use some of the intermediate results to obtain upper and lower bounds that appear to be new, in the purely Bayesian setting, on the optimal filtering performance in terms, respectively, of the mutual information between the noise-free and noisy sequence, and the entropy of the noise-free sequence causally conditioned on the noisy one. 
\end{abstract}


\section{Introduction}

In the discrete filtering\footnote{Denoising generally refers to an estimation problem in a fully non-causal setting, where the entire data sequence is available for estimating the noise-free data at each time step. Filtering, on the other hand, involves causal estimation, performed ``on the fly'' with potential delay or lookahead. In this paper, we use the term ``filtering'' to encompass all problems that estimate noise-free from noisy data.} problem, the objective is to estimate the underlying noise-free data from its noisy counterpart, which has been corrupted by a discrete memoryless channel, both having a finite alphabet. When the distribution of the noise-free data and the statistics of the channel are known, it is possible to derive the optimum Bayesian scheme. However, our focus is on universal schemes that function effectively without prior knowledge of the noise-free data distribution, which is generally unavailable in practice.

Universal estimation of discrete data has been extensively explored. From prediction, where the objective is to estimate the next symbol based on the past sequence \cite{feder1992universal, merhav1998universal, weinberger2002delayed, weissman2003competitive, cesabianchi1999prediction} to noisy prediction problems corresponding to filtering with a delay of one \cite{weissman2001universal, weissman2004universal}. Other related settings include fully non-causal schemes that estimate the noise-free data using the entire noisy data \cite{weissman2005universal,1715526,1405281}, and causal filters that rely only on the noise-free data available up to the current time step \cite{weissman2007Filtering, moon2009universal, 4439859}. These and related problems have been considered in many parts of the literature beyond information theory as well. We refer to   \cite{merhav1998universal, cesabianchi1999prediction, weissman2005universal} and references therein for a more extensive account.  

Our main contribution lies in extending the universal discrete filtering framework to accommodate arbitrary fixed lookahead and delay. This extension is applicable to a wide range of real-world scenarios. Given any universal sequential probability assignment (SPA) and any discrete memoryless channel with an invertible channel matrix, our scheme can filter sequences from stochastic sources and be guaranteed to asymptotically (in the limit of large amounts of data) attain optimum performance. Notably, with computationally efficient SPAs, such as those based on LZ78, our filter can be deployed in low-computation environments.

We define the key elements used throughout the paper in Section~\ref{problem_setting}.  We derive the expressions for the Bayes optimal filters given the true distributions and establish bounds on the excess estimation loss when the filter is optimized for a distribution other than the true one in Section~\ref{bounding_excess_loss}. In passing, in Subsection \ref{subsec: mutual information bounds}, we use these bounds to establish what appear to be new upper bounds in the purely Bayesian setting on the optimal filtering performance in terms of the mutual information between the noise-free and noisy sequence.  We also derive lower bounds in terms of the entropy of the noise-free process causally conditioned \cite{6555871} on the noisy one. 
In Section~\ref{universal_filtering_schemes}, we discuss the tightness of these bounds and use them to derive universal filters for three settings: causal, delayed, and non-causal with lookahead. Section~\ref{implementation} is dedicated to implementation aspects pertaining to these universal filters, followed by experimental results using universal filters derived specifically from LZ78-based SPAs in Section~\ref{experiment}. For readability, the main body of the paper presents the key results along with proof sketches, with full proofs deferred to the Appendix. We conclude in Section \ref{sec: conclusions} with a summary of our main contributions and future directions. 

\section{Problem Setting} \label{problem_setting}
\subsection{Notation and Definitions}
\begin{itemize}

    \item{For random variable $X$ taking values in a finite alphabet $\Xcal$, $P_X$ denotes the $|\Xcal|$-vector of probabilities $\Pbb_X(\cdot)$. Similarly, for random variable $Y$, $P_{X|Y=y}$ denotes the $|\Xcal|$-vector of probabilities $\Pbb_{X|Y}(\cdot|Y=y)$. $P_{X|Y}$ is a random column vector for PMF of $X$, where the randomness comes from $Y$.}

    \item{We denote the alphabet of the noise-free signal as $\Acal_X$, that of the noisy signal as $\Acal_Z$, and that of the filtered signal as $\Acal_{\hat{X}}$. }
\end{itemize}

We define the core concepts underlying this paper: 
\tikzstyle{block} = [draw, fill=white, rectangle, 
    minimum height=3em, minimum width=6em]

\tikzstyle{input} = [coordinate]
\tikzstyle{output} = [coordinate]
\tikzstyle{pinstyle} = [pin edge={to-,thin,black}]
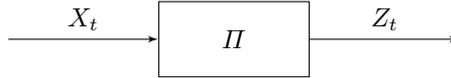
\begin{figure}[h]
    \centering
    \begin{tikzpicture}[auto, node distance=2cm,>=latex']
        \node[draw, rectangle, minimum width=2cm, minimum height=1cm] (block) {$\Pi$};
        \draw[->] (-3,0) -- (block.west) node[midway, above] {$X_t$};  
        \draw[->] (block.east) -- (3,0) node[midway, above] {$Z_t$};
    \end{tikzpicture}
    \caption{A memoryless noisy channel}
    \label{fig:noisy_channel}
\end{figure}

For this work, we model the noise process as a discrete memoryless channel, as shown in Figure \ref{fig:noisy_channel}. We aim to optimally estimate $X^n = (X_1, \ldots, X_n)$, given the noisy signal $Z^n$. The noisy channel is characterized by a channel matrix $\Pi(x, z) \coloneqq \mathbb{P}\left(Z_t=z|X_t=x\right)$ where $x$ indexes the rows and $z$ indexes the columns, and we assume that $\Pi$ is invertible\footnote{This is a standard and benign assumption aligning with previous literature \cite{weissman2005universal}}. We define a loss function $\mathrm{\Lambda}:\Acal_X \times \Acal_{\hat{X}} \rightarrow [0,\infty)$ and its matrix representation $\mathrm{\Lambda}=\{\mathrm{\Lambda}(i,j)\}_{i\in\Acal_X, j\in\Acal_{\hat{X}}}$, where $\mathrm{\Lambda}(i,j)$ represents the loss if the symbol $i$ is estimated as the symbol $j$. $\mathrm{\Lambda}_\text{max}$ denotes $\max_{i\in\Acal_X, j\in\Acal_{\hat{X}}} \mathrm{\Lambda}(i,j)$, and $\lambda_k$ denotes the $k$-th column of $\mathrm{\Lambda}$. 

    

\subsection{Properties}
We can obtain the expression for marginal probabilities of $X_t$ using the following steps. By the law of total probability,
\begin{equation}
    \Pbb\left(Z_t=z\right) = \sum_{x \in \Acal_X} \Pbb\left(Z_t=z|X_t=x\right) \Pbb\left(X_t=x\right).
\end{equation}
This can be formulated as the matrix-vector product
\begin{equation}
    P_{Z_t} = \Pi^\top P_{X_t} \iff P_{X_t} = \Pi^{-\top} P_{Z_t}.
\end{equation}

Using this expression for $P_{X_t}$, we can derive the conditional probabilities of $X_t$ given $Z_t$. By Bayes' theorem,
\begin{equation}
    \Pbb\left(X_t=x|Z_t=z\right) = \frac{\Pbb\left(Z_t=z|X_t=x\right) \Pbb\left(X_t=x\right)}{\Pbb\left(Z_t=z\right)} = \frac{\Pi(x, z) \Pbb\left(X_t=x\right)}{\Pbb\left(Z_t=z\right)}.
\end{equation}
In matrix-vector form,
\begin{equation} \label{conditional_probability}
    P_{X_t|Z_t=z} = \frac{\Pi(\cdot, z) \odot \left(\Pi^{-\top} P_{Z_t}\right)}{\Pbb\left(Z_t=z\right)} \triangleq F\left(P_{Z_t}, \Pi, z\right),
\end{equation}
where $\odot$ represents element-wise multiplication (\ie, Hadamard product or Schur product).

\begin{lemma}[Conditional Probabilities of $X_t$ given the sequence $Z^t$]
\label{conditionalcausal}
For all $1 \leq t \leq n$,
    \[P_{X_t|Z^t=z^t} = F\left(P_{Z_t|z^{t-1}}, \Pi, z_t\right),\]
and
\[P_{X_t|Z^n=z^n} = F\left(P_{Z_t|z^{t-1}, z_{t+1}^n}, \Pi, z_t\right).\]
\end{lemma}
We can extend \eqref{conditional_probability} to $P_{X_t|Z^t=z^t}$ or $P_{X_t|Z^n=z^n}$ by plugging in $P_{Z_t|Z^{t-1}=z^{t-1}}$ or $P_{Z_t|Z^{t-1}=z^{t-1}, Z_{t+1}^n=z_{t+1}^n}$, respectively, for $P_{Z_t}$ in the equation $P_{X_t|Z_t=z} = F\left(P_{Z_t}, \Pi, z\right)$. This is then followed by applying Bayes theorem and the conditional independence between $Z_t$ and other $Z$'s given $X_t$, which is due to the memorylessness of the channel. 

\section{Bounding Excess Loss} \label{bounding_excess_loss}
In the context of estimation, the optimal decision minimizing the expected loss with respect to a given loss function $\mathrm{\Lambda}(\cdot,\cdot)$ and a distribution $P_X$ is referred to as the \textbf{Bayes response}. The action space comprises estimates of the random variable $X\sim P_X$, and the optimal action/estimate is the \textbf{Bayes estimate} of $X$ or the \textbf{Bayes response} to $P_X$.

\begin{definition}[Bayes estimator/response]
The Bayes estimator of $X$ under the distribution $P_X$ and loss function $\mathrm{\Lambda}(\cdot,\cdot)$ is given by:
\[\hat{X}_B\left(P_X\right) = \arg \min_{\hat{x}} \Ebb[\mathrm{\Lambda}(X,\hat{x})] = \arg \min_{\hat{x}} \sum_{x\in\Acal_X} \mathrm{\Lambda}(x, \hat{x}) P_X(x) = \arg \min_{\hat{x}} \lambda_{\hat{x}}^T P_X.\]
\end{definition}
Note that $\hat{X}_B\left(P_X\right)$ can be generalized to $\hat{X}_B(\mathbf{v})$ with arbitrary vector $\mathbf{v}$. This $\mathbf{v}$ can also be a random vector such as $P_{X|Z}$.

\begin{lemma}[Optimal Estimator]\label{lem:optimal_estimator}
Let $W \in \Wcal$ be a random variable. Then, the best estimator of $X$ based on $W$ in the sense of minimizing the expected loss can be represented as follows:
    \[\arg\min_{\hat{X}(\cdot):\Wcal \rightarrow \Acal_{\hat{X}}} \Ebb \left[ \mathrm{\Lambda} \left(X, \hat{X}(W) \right) \right] = \hat{X}_B \left(P_{X|W}\right).\]
   \end{lemma}
   The proof is merely an application of the law of total expectation and the definition of the Bayes estimator. 

\begin{definition}[Sequential Probability Assignment, SPA]
    A Sequential Probability Assignment (SPA) assigns probabilities to the next symbol in a sequence based on all previous symbols. For a sequence $X$, the SPA $q$ is a collection of conditional probability distributions:
    \[q \triangleq \left\{q_t(x_t|x^{t-1})\right\}_{t \geq 1},\]
    where $q_t(\cdot|x^{t-1})$ is a probability distribution over the possible values of $x_t$.
\end{definition}

Here, we establish theoretical bounds on excess expected estimation loss beyond the Bayes optimum due to optimizing the filters to an SPA other than the true one. These will subsequently enable the design of universal filtering schemes with robust performance across diverse settings.
\subsection{Causal Estimation}
In the causal setting, we estimate $X_{t}$ given the past and current observations $Z^t$, as in \cite{weissman2007Filtering}. 
By Lemma~\ref{lem:optimal_estimator}, we know that the optimal causal filter given $Z^t$ is $\hat{X}_B (P_{X_t|Z^t})$, where $P_{X_t|Z^t}$ can be calculated from $P_{Z_t|Z^{t-1}}$ using Lemma~\ref{conditionalcausal}. Thus, we can derive the expression for the optimal causal filter $\hat{X}_{t}^\text{opt}\left(Z^t\right)$.
\begin{equation}\label{eq:X_opt}
    \hat{X}_{t}^\text{opt}\left(Z^t\right) = \hat{X}_{B}\left(P_{X_{t}|Z^{t}}\right) =  \hat{X}_{B}\left(F\left(P_{Z_t|Z^{t-1}}, \Pi, Z_t\right)\right).
\end{equation}
However, the ground truth sequential probability $P_{Z_t|Z^{t-1}}$ is almost always inaccessible in real world problems. Therefore, instead of obtaining $P_{Z_t|Z^{t-1}}$, if we have a good SPA $Q_{Z_t|Z^{t-1}}$, then we can estimate the distribution of $X_t$ as
\begin{equation}
    \hat{P}_{X_t|Z^t} = F\left(Q_{Z_t|Z^{t-1}}, \Pi, Z_t\right).\label{eqn:causal-dist-estimate}
\end{equation}
We denote the estimator corresponding to $\hat{P}_{X_t|Z^t}$ as follows:
\begin{equation} \label{eq:causal_estimator}
    \hat{X}_B^{Q_{Z_t|Z^{t-1}}}\left(Z_t\right) \triangleq \hat{X}_B\left(F\left(Q_{Z_t|Z^{t-1}}, \Pi, Z_t\right)\right).
\end{equation}
Note that $\hat{X}_B^{P_{Z_t|Z^{t-1}}}\left(Z_t\right)$ is equal to $\hat{X}_t^\text{opt}\left(Z^t\right)$ in (\ref{eq:X_opt}), the optimal estimator that minimizes the expected loss.

\subsection{Mismatched Estimation} \label{subsec: mismatched bounding_excess_loss}
The expected excess loss due to mismatch can be bounded as follows:

\begin{theorem}[Excess Loss Bound for Causal Estimation]
\label{thm:normalizedloss} For a sequence $Z^{n}$ and SPA $Q$, the expected difference in normalized loss is bounded by: 
    \[\Ebb\left[ \frac{1}{n} \sum_{t=1}^{n}\mathrm{\Lambda}\left(X_t, \hat{X}_B^{Q_{Z_t|Z^{t-1}}}\left(Z_t\right)\right) - \frac{1}{n} \sum_{t=1}^{n} \mathrm{\Lambda}\left(X_t, \hat{X}_t^\text{opt}\left(Z^t\right)\right)\right] \leq \sqrt{2} C_1(\Pi) \mathrm{\Lambda}_\text{max} \sqrt{\frac{1}{n}D\left(P_{Z^n}\Vert Q_{Z^n}\right)}, \]
where $C_1(\Pi)=\left|\Pi^{-T}\right|_\text{max} \left| \Acal_X \right|$ 
and $\left|\Pi^{-T}\right|_\text{max}$ denotes the maximum component of $\left|\Pi^{-T}\right|$. $D(\cdot\Vert\cdot)$ denotes the relative entropy (KL divergence).
    \end{theorem}
The proof is based on noting that, for each $t$, the difference in losses between two cases -- one using the estimator based on the ground truth sequential probability $P_{Z_t|Z^{t-1}}$ and the other using its estimate $Q_{Z_t|Z^{t-1}}$ -- can be bounded using a ``single-letter'' bound on their $\ell_1$ distance (Lemma~\ref{lem:loss_bound}). We then use Pinsker's inequality to bound that distance with square root of the relative entropy, then use Jensen's inequality to put the summation into the square root. Finally, we use the chain rule for relative entropy. This result provides a distribution-agnostic bound for the excess loss in any mismatched estimation scenario. The tightness of the bound is discussed in Section~\ref{tightness}.

\subsubsection{Estimation with Delay}

In many real world scenarios, data cannot be processed instantly or signal transmission introduces lag. This introduces a delay $d \geq 1$, and we may have to estimate $X_t$ based on the delayed noisy sequence $Z^{t-d}$. Again, by Lemma~\ref{lem:optimal_estimator}, we can find the Bayes optimal filter given $Z^{t-d}$, denoted as $\hat{X}_{t}^\text{opt} \left(Z^{t-d}\right)$. 

\begin{equation}
    \hat{X}_{t}^\text{opt} \left(Z^{t-d}\right) = \hat{X}_{B}\left(P_{X_{t}|Z^{t-d}}\right)
\end{equation}

If we assume that we have $P_{Z_t|Z^{t-d}}$, we can derive $P_{X_t|Z^{t-d}}$ as 

\begin{lemma} \label{lem:conditionaldelay}
For $d \geq 1$,
    \[P_{X_t|Z^{t-d}}=\Pi^{-T} P_{Z_t|Z^{t-d}}\]
\end{lemma}
The proof merely uses marginalization and the fact that $Z_t - X_t - Z^{t-d}$.

With $Q_{Z_t|Z^{t-d}}$ induced by SPA $Q_{Z_t|Z^{t-1}}$, the corresponding estimator can be represented as follows. 
\begin{equation} \label{eq:causal_estimator_delay}
    \hat{X}_B^{Q_{Z_t|Z^{t-d}}} \left(Z_t \right) \triangleq \hat{X}_B \left( \Pi^{-T}Q_{Z_t|Z^{t-d}} \right)
\end{equation}

The excess estimation loss bound for this setting is:
\begin{theorem}[Excess Loss Bound for Delay]\label{thm:normalizedloss_delay}
    For $d \geq 1$,
    \[\Ebb  \left[ \frac{1}{n} \sum_{t=1}^{n}\mathrm{\Lambda}\left(X_t, \hat{X}_B^{Q_{Z_t|Z^{t-d}}}\left(Z_t\right) \right)  - \frac{1}{n} \sum_{t=1}^{n} \mathrm{\Lambda}\left(X_t, \hat{X}_t^\text{opt}\left(Z^{t-d}\right)\right) 
    \right] \leq \sqrt{2} \mathrm{\Lambda}_\text{max} C_1(\Pi) \sqrt{\frac{d}{n}D\left(P_{Z^n}\Vert Q_{Z^n}\right)}. \]
\end{theorem}
The proof is similar to that of Theorem~\ref{thm:normalizedloss}, but we first use Lemma~\ref{lem:loss_bound_2} to bound the difference in losses between the two cases -- one using the estimator based on the ground truth sequential probability $P_{Z_t|Z^{t-d}}$ and the other using its estimate $Q_{Z_t|Z^{t-d}}$.
Note that the delay manifests in the bound via an additional factor $d$  under the square root. This does not necessarily imply an intrinsic penalty in the excess loss itself. 
\subsubsection{Estimation with Lookahead}
In applications that prioritize accuracy over real-time requirements, some filtering problems allow for some lookahead $l \geq 0$. Beyond the use of a limited lookahead $l$, bidirectional modeling approaches, such as those discussed in \cite{1715526,1405281}, extend to consider the entire context, leveraging both past and future observations for enhanced filtering performance. Under the setting of limited lookahead $l$, the filtering problem can be framed as estimating $X_{t}$ based on the past sequence $Z^{t-1}$, current symbol $Z_t$, and a lookahead sequence $Z_{t+1}^{t+l}$. In such non-causal scenarios, we can also derive the optimal filter $\hat{X}_t^\text{opt}\left(Z^{t+l}\right)$ as follows:
\begin{equation}
    \hat{X}_t^\text{opt}\left(Z^{t+l}\right)=\hat{X}_B^{P_{Z_t|Z^{t-1}, Z_{t+1}^{t+l}}}\left(Z_t\right).
\end{equation}

This can be shown using the result from Lemma~\ref{lem:optimal_estimator}. In this case we are conditioning $Z_t$ on $Z^{t-1}$ and $Z_{t+1}^{t+l}$ instead of merely $Z^{t-1}$.

Given $Q_{Z_t|Z^{t-1},Z_{t+1}^{t+l}}$ induced by SPA $Q$, the estimator is:
\begin{equation} \label{eq:non_causal_estimator}
    \hat{X}_B^{Q_{Z_t|Z^{t-1},Z_{t-1}^{t+l}}}\left(Z_t\right) \triangleq \hat{X}_B\left(F\left(Q_{Z_t|Z^{t-1},Z_{t-1}^{t+l}}, \Pi, Z_t\right)\right).
\end{equation}

The excess estimation loss is bounded as follows:

\begin{theorem}[Excess Loss Bound for Lookahead]
\label{thm:normalizedloss_lookahead}
For $l \geq 0$,
\begin{align*} 
&\Ebb\left[ \frac{1}{n} \sum_{t=1}^{n}\mathrm{\Lambda}\left(X_t, \hat{X}_B^{Q_{Z_t|Z^{t-1},Z^{t+l}_{t+1}}}\left(Z_t\right)\right) - \frac{1}{n} \sum_{t=1}^{n} \mathrm{\Lambda}\left(X_t, \hat{X}_t^\text{opt} \left(Z^{t+l}\right)\right)\right] \\
&\leq \sqrt{2} C_1(\Pi) \mathrm{\Lambda}_\text{max} \sqrt{\frac{l+1}{n}D\left(P_{Z^{n+l}}\Vert Q_{Z^{n+l}}\right)}
\end{align*}
\end{theorem}
The proof is similar to that of Theorem~\ref{thm:normalizedloss}, using chain rules and data processing properties of relative entropy. The derived excess loss bound scales with the lookahead $l$ in a manner similar to how it scales with the delay in the delayed case. 
By extending the lookahead from a predefined $l$ to the entire dataset, a different bound was obtained in \cite{4655473}, Equation (287) in terms of the ``erasure divergence''. 
The bound in Theorem \ref{thm:normalizedloss_lookahead} is tight enough to deduce that the excess loss converges to zero whenever $\lim_{n\to\infty} \frac{1}{n} D\left(P_{Z^n} \Vert Q_{Z^n}\right) = 0$ for fixed $\ell$, but not for $\ell$ growing with $n$. The existence of bounds that would allow $\ell$ to grow with $n$ at some rate while still vanishing when $\lim_{n\to\infty} \frac{1}{n} D\left(P_{Z^n} \Vert Q_{Z^n}\right) = 0$, and the relationship to the erasure divergence bound of \cite{4655473},  are left for future work.

\subsection{Information Theoretic Bounds on the Bayes Optimal Performance} \label{subsec: mutual information bounds}
\begin{corollary}
\label{cor:normalizedloss} Assume $\mathrm{\Lambda}$ is such that $\forall x \in \mathcal{A}_X$ $\exists \hat{x} \in \mathcal{A}_{\hat{X}}$ with $\mathrm{\Lambda} (x, \hat{x})=0$. 
For any individual sequence $x^{n}$ and SPA $Q$ we have  
    \[\Ebb\left[ \frac{1}{n} \sum_{t=1}^{n}\mathrm{\Lambda}\left(x_t, \hat{X}_B^{Q_{Z_t|Z^{t-1}}}\left(Z_t\right)\right) \right] \leq \sqrt{2} C_1(\Pi) \mathrm{\Lambda}_\text{max} \sqrt{\frac{1}{n}D\left(P_{Z^n|x^n}\Vert Q_{Z^n}\right)}, \]
where the expectation on the left side is with respect to the channel noise.
    \end{corollary}
    This corollary follows directly by specializing Theorem \ref{thm:normalizedloss} to $P = \delta_{x^n}$ and noting that in this case $\Ebb\left[  \frac{1}{n} \sum_{t=1}^{n} \mathrm{\Lambda}\left(x_t, \hat{X}_t^\text{opt}\left(Z^t\right)\right)\right] = 0$ (thanks to the additional stipulation on $\mathrm{\Lambda}$). This corollary is key in establishing the following: 
\begin{theorem}[Mutual information upper bound on the optimal performance]
\label{thm:mutinfo normalizedloss} Assume $\mathrm{\Lambda}$ is such that $\forall x \in \mathcal{A}_X$ $\exists \hat{x} \in \mathcal{A}_{\hat{X}}$ with $\mathrm{\Lambda} (x, \hat{x})=0$. The expected normalized loss of the optimal filter satisfies:  
    \[\Ebb\left[ \frac{1}{n} \sum_{t=1}^{n} \mathrm{\Lambda}\left(X_t, \hat{X}_t^\text{opt}\left(Z^t\right)\right)\right] \leq \sqrt{2} C_1(\Pi) \mathrm{\Lambda}_\text{max} \sqrt{\frac{1}{n}I (X^n ; Z^n)}.  \]
    \end{theorem}
The bound in Theorem \ref{thm:mutinfo normalizedloss} may seem counter-intuitive at first glance since the less informative $Z^n$ is about $X^n$ (smaller mutual information) the worse (larger expected loss) one might expect the filtering performance to be. Note, however, that for a fixed distribution on the noise-free process $X^n$, the more noisy the channel (smaller mutual information) the larger the constant $C_1(\Pi)$ will tend to be.      For a fixed channel, lower mutual information may be due to lower entropy of $X^n$, and it would make intuitive sense that the smaller the entropy of $X^n$ the smaller expected filtering loss can be attained.          
\begin{proof}
    \begin{align*}
    \Ebb\left[ \frac{1}{n} \sum_{t=1}^{n} \mathrm{\Lambda}\left(X_t, \hat{X}_t^\text{opt}\left(Z^t\right)\right)\right] &= \Ebb \left[ \frac{1}{n} \sum_{t=1}^{n} \mathrm{\Lambda}\left(X_t, \hat{X}_B^{P_{Z_t|Z^{t-1}}}\left(Z^t\right)\right)\right] \\
    &= \sum_{x^n} \Ebb \left[ \frac{1}{n} \sum_{t=1}^{n} \mathrm{\Lambda}\left(X_t, \hat{X}_B^{P_{Z_t|Z^{t-1}}}\left(Z^t\right)\right) \Big{|} X^n = x^n \right] P_{X^n}(x^n) \\
    &\leq \sum_{x^n} \sqrt{2} C_1(\Pi) \mathrm{\Lambda}_\text{max} \sqrt{ \frac{1}{n} D\left(P_{Z^n|x^n} \Vert P_{Z^n} \right)} P_{X^n} (x^n) \qquad (\because \text{ Corollary \ref{cor:normalizedloss}}) \\
    &\leq \sqrt{2} C_1(\Pi) \mathrm{\Lambda}_\text{max} \sqrt{\frac{1}{n} \sum_{x^n} D\left(P_{Z^n|x^n} \Vert P_{Z^n} \right) P_{X^n} (x^n)} \qquad (\because \text{ Jensen's inequality}) \\
    &= \sqrt{2} C_1(\Pi) \mathrm{\Lambda}_\text{max} \sqrt{\frac{1}{n} I(X^n ; Z^n)}
    \end{align*}
\end{proof}

Using analogous arguments, the following result can be derived by applying Theorem \ref{thm:normalizedloss_delay} and Theorem \ref{thm:normalizedloss_lookahead} in place of Theorem \ref{thm:normalizedloss} for the cases of filtering with delay and lookahead, respectively. Additionally, Theorem \ref{thm:mutinfo normalizedloss} is invoked to complete the derivation.
\begin{theorem}
\label{thm:mutinfo normalizedloss with delay or lookahead} Assume $\mathrm{\Lambda}$ is such that $\forall x \in \mathcal{A}_X$ $\exists \hat{x} \in \mathcal{A}_{\hat{X}}$ with $\mathrm{\Lambda} (x, \hat{x})=0$. Then   for $d \geq 1$,
    \[\Ebb\left[ \frac{1}{n} \sum_{t=1}^{n} \mathrm{\Lambda}\left(X_t, \hat{X}_t^\text{opt}\left(Z^{t-d} \right)\right)\right] \leq \sqrt{2} C_1(\Pi) \mathrm{\Lambda}_\text{max} \sqrt{\frac{d}{n}I (X^n ; Z^n)}  \]
and, for $l \geq 0$,
\[\Ebb\left[ \frac{1}{n} \sum_{t=1}^{n} \mathrm{\Lambda}\left(X_t, \hat{X}_t^\text{opt}\left(Z^{t + l} \right)\right)\right] \leq \sqrt{2} C_1(\Pi) \mathrm{\Lambda}_\text{max} \sqrt{\frac{l+1}{n}I (X^n ; Z^{n+l})}  . \]
    \end{theorem}
We also have the following information theoretic lower bounds on these quantities.  \begin{theorem}[Information-theoretic lower bounds on the Bayes optimal performances]\label{thm:mutual-information-lower-bounds}
    Suppose that $\mathcal{A}_X = \mathcal{A}_{\hat{X}} = \{0, 1, \dots, m-1\}$.
    Also, suppose the loss function is subtractive, i.e., $\mathrm{\Lambda}(x, x') = \rho(x - x')$ for some function $\rho$ with $\rho(0) = 0$, where subtraction is performed modulo $m$.

    Then, for any fixed integer $l$,
    \[\mathbb{E}\left[ \frac{1}{n} \sum_{t=1}^n \mathrm{\Lambda}\left(X_t,  \hat{X}^{opt}\left(Z^{t+l}\right)\right)\right] \geq \phi^{-1}\left( \tfrac{1}{n} H(X^n\lVert Z^{n+l}) \right),\]
     where $\phi(\cdot)$ is defined as
    \[\phi(D) \triangleq \max \left\{ H(U) : U \text{ random variable over } \mathcal{A}_X \text{ s.t. } \mathbb{E}[\rho(U)] \leq D \right\},\]
    and $H(X^n\lVert Z^{n+l})$ is the entropy of $X^n$ causally conditioned on $Z^{n+l}$,
    \[H(X^n\lVert Z^{n+l}) = \sum_{t=1}^n H(X_t|X^{t-1}, Z^{t+l}).\]
\end{theorem}
$\phi(D)$ is strictly increasing for $0 \leq D \leq \frac{1}{m}\sum_{i=1}^{m-1}\rho(i)$, so it suffices to show that $\phi$, evaluated at the average expected loss of $\hat{X}^{opt}$, is lower-bounded by $H(X^n\lVert Z^{n+l})$.
As $\phi(D)$ is concave, this lower bound can be obtained via Jensen's inequality, the definition of $\phi(D)$, and the fact that conditioning reduces entropy.
The full proof is in the Appendix.

\begin{corollary} \label{cor:mutual-information-lower-bounds}
    Suppose the conditions on the alphabet and loss function from Theorem \ref{thm:mutual-information-lower-bounds} hold.
    Then, for the optimal denoiser that has access to the full sequence $Z^n$,
    \[\mathbb{E}\left[ \frac{1}{n} \sum_{t=1}^n \mathrm{\Lambda}\left(X_t,  \hat{X}^{opt}\left(Z^{n}\right)\right)\right] \geq \phi^{-1}\left( \tfrac{1}{n} H(X^n| Z^{n}) \right).\]
\end{corollary}
\begin{proof}
    Plugging in $l = n$ in Theorem \ref{thm:mutual-information-lower-bounds},
    \[\mathbb{E}\left[ \frac{1}{n} \sum_{t=1}^n \mathrm{\Lambda}\left(X_t,  \hat{X}^{opt}\left(Z^{n}\right)\right)\right] \geq \phi^{-1}\left( \frac{1}{n} \sum_{t=1}^n H(X_t|X^{t-1},Z^n)\right) = \phi^{-1}\left(\tfrac{1}{n}H(X^n|Z^n)\right).\]
\end{proof}
In Section \ref{experiment} we evaluate the optimal expected filtering losses along with these lower bounds for a specific noise-free source and channel.

\section{Universal Filtering Schemes} \label{universal_filtering_schemes}
In the previous section, we obtained bounds on the excess loss due to the mismatch between the true and the SPA assumed by the filter. The bounds are in terms of the relative entropy. If the SPA assumed by the filter satisfies $\lim_{n\to\infty} \frac{1}{n} D\left(P_{Z^n} \Vert Q_{Z^n}\right) = 0$, the excess loss converges to zero, ensuring the universality of the induced estimators under universal SPAs.

\subsection{Universality with respect to limited classes} \label{tightness}

The bounds derived in Section~\ref{bounding_excess_loss} are order-tight in various senses we now briefly discuss.  For concreteness, we focus on the bound and setting of Theorem~\ref{thm:normalizedloss}, with the implication that analogous points can be made for those of Theorem~\ref{thm:normalizedloss_delay} and Theorem \ref{thm:normalizedloss_lookahead}.

\subsubsection{Finite Uncertainty Set}
We first consider a finite uncertainty set of $N$ different sources. Let $\mathbf{X}$ be governed by a law from the finite  uncertainty set $\left\{ P_\mathbf{X}^{(i)} \right\}_{i=1}^N$, and let $\left\{ P_\mathbf{Z}^{(i)}\right\}$ be the respective laws of the noisy observation process.  We can define $Q_\mathbf{Z}$ to be the uniform mixture:
\[Q_\mathbf{Z} = \frac{1}{N} \sum_{i=1}^N P_\mathbf{Z}^{(i)}, \]
for which 
\[D \left(P_{Z^n}^{(i)} \Vert Q_{Z^n} \right) \leq \log N, \quad \forall 1 \leq i \leq N.\]
Applying this to the bound from Theorem~\ref{thm:normalizedloss} yields
\begin{align*}
    \Ebb \left[ \frac{1}{n} \sum_{t=1}^{n}\mathrm{\Lambda}\left(X_t, \hat{X}_B^{Q_{Z_t|Z^{t-1}}}\left(Z_t\right)\right) - \frac{1}{n} \sum_{t=1}^{n} \mathrm{\Lambda}\left(X_t, \hat{X}_t^\text{opt}\left(Z^t\right)\right)\right] &\leq \sqrt{2} C_1(\Pi) \mathrm{\Lambda}_\text{max} \sqrt{\frac{1}{n}D\left(P_{Z^n}^{(i)}\Vert Q_{Z^n}\right)} \\
    &\leq \sqrt{2} C_1(\Pi)\mathrm{\Lambda}_\text{max} \sqrt{\frac{\log N}{n}} \qquad \forall 1 \leq i \leq N , 
\end{align*}
where the expectation on the left side is under any of the sources in the uncertainty set. 
The $\sqrt{1/n}$ behavior of the bound is order tight, as shown similarly in, e.g., Proposition 3 of \cite{6555871}. The necessity of a small $\frac{\log N}{n}$ for universality follows from channel coding arguments analogous to those in \cite{weissman2003competitive}.   


\subsubsection{Parametric Uncertainty Set}
Consider now an uncertainty set $\left\{ P_\mathbf{X}^\theta \right\}_{\theta \in \Gamma}$
with $\Gamma \subset \mathbb{R}^k$ compact and the parametrization satisfying the standard smoothness conditions of, e.g., \cite{rissanen1984universal}. Let  $\left\{ P_\mathbf{Z}^\theta \right\}_{\theta \in \Gamma}$ be the set induced from $\left\{ P_\mathbf{X}^\theta \right\}_{\theta \in \Gamma}$ using $\Pi$. The minimum description length (MDL) framework 
\cite{rissanen1984universal, merhav1995strong, merhav1998universal} implies existence of  $Q_\mathbf{Z}$ satisfying
\[\frac{1}{n} D \left( P_{Z^n}^\theta \Vert Q_{Z^n} \right) \leq \frac{k}{2} \frac{\log n}{n} (1+\epsilon), \qquad \forall \theta \in \Gamma\]
for arbitrarily small $\epsilon$ and all sufficiently large $n$. Similar to the case for a finite uncertainty set, we can apply this inequality to the bound from Theorem~\ref{thm:normalizedloss}: 
\begin{align*}
    \Ebb\left[ \frac{1}{n} \sum_{t=1}^{n}\mathrm{\Lambda}\left(X_t, \hat{X}_B^{Q_{Z_t|Z^{t-1}}}\left(Z_t\right)\right) - \frac{1}{n} \sum_{t=1}^{n} \mathrm{\Lambda}\left(X_t, \hat{X}_t^\text{opt}\left(Z^t\right)\right)\right] &\leq \sqrt{2} C_1(\Pi) \mathrm{\Lambda}_\text{max} \sqrt{\frac{1}{n}D\left(P_{Z^n}^\theta \Vert Q_{Z^n}\right)} \\
    &\leq \sqrt{2} C_1(\Pi)\mathrm{\Lambda}_\text{max} \sqrt{\frac{k}{2} \frac{\log n}{n} (1+\epsilon)} \qquad \forall \theta \in \Gamma
\end{align*}
where the expectation on the left side is under any of the sources in the uncertainty set. 


\subsection{Universality with respect to Stationary Sources}
Building on the bounds in the previous section, we now explore universal filtering schemes for the causal, delayed, and non-causal settings. These schemes leverage universal SPAs. First, we highlight the existence of universal SPAs $Q$ satisfying
\begin{equation} \label{eq:universal_SPA}
    \lim_{n\to\infty} \frac{1}{n} D\left(P_{Z^n} \Vert Q_{Z^n}\right) = 0 \quad \forall P \text{ stationary}
\end{equation}
In principle, any universal compressor would induce a universal SPA, including Context Tree Weighting (CTW)\cite{382012} and Prediction by Partial Matching (PPM)\cite{1096090}. Some are more computationally efficient, such as the LZ78-based SPA, as elaborated in Section~\ref{experiment}.



 Given $Q$ satisfying (\ref{eq:universal_SPA}) , we define the asymptotically optimal universal filter $\hat{X}^{ao}_t\left(Z^t\right)$ as:
\begin{equation}\label{eq:Asymptotic_filter}
\hat{X}^{ao}_t\left(Z^t\right) \triangleq \hat{X}_B\left(F\left(Q_{Z_t|Z^{t-1}}, \Pi, Z_t\right)\right) =  \arg \min_{\hat{X}(\cdot):\Acal_Z^t \rightarrow \Acal_{\hat{X}}} \sum_{x\in\Acal_X} \mathrm{\Lambda}\left(x, \hat{X}\left(Z^t\right)\right) \frac{\Pi\left(x, Z_t\right) \left(\Pi^{-T}(x,\cdot) Q_{Z_{t}|Z^{t-1}}\right)}{ Q_{Z_t|Z^{t-1}}\left(Z_t|Z^{t-1}\right)}.
\end{equation}

E.g., when the loss function $\mathrm{\Lambda}$ is squared error, it reduces to:

\begin{equation}
    \hat{X}^{ao}_t\left(Z^t\right) = \Ebb_{X \sim F\left(Q_{Z_t|Z^{t-1}}, \Pi, Z_t\right)} \left[ X \right] = \sum_{x\in\Acal_X} x \cdot \frac{\Pi\left(x, Z_t\right) \left(\Pi^{-T}(x,\cdot) Q_{Z_{t}|Z^{t-1}}\right)}{ Q_{Z_t|Z^{t-1}}\left(Z_t|Z^{t-1}\right)}.
\end{equation}

Similar to \eqref{eq:Asymptotic_filter}, we can derive the asymptotically optimal universal filter in the case of delay $d \geq 1$:
\begin{align}
\hat{X}^{ao}_{t}\left(Z^{t-d}\right) &\triangleq \hat{X}_B \left( Q_{X_t|Z^{t-d}} \right) \nonumber\\
    &= \hat{X}_B \left( \Pi^{-T} Q_{Z_t|Z^{t-d}} \right)  \nonumber\\
    &= \arg \min_{\hat{X}(\cdot):\Acal_Z^{t-d} \rightarrow \Acal_{\hat{X}}} \sum_{x\in\Acal_X} \mathrm{\Lambda}\left(x, \hat{X}\left(Z^{t-d}\right)\right) \Pi^{-T} Q_{Z_t|Z^{t-d}} (x). \label{eq: universal filter with delay}
\end{align}

From Theorem~\ref{thm:normalizedloss_lookahead}, we can similarly derive the asymptotically optimal universal filter with a lookahead of $l \geq 0$ as: 

\begin{align}
\hat{X}^{ao}_{t}\left(Z^{t+l}\right) &\triangleq \hat{X}_B\left(F\left(Q_{Z_t|Z^{t-1},Z^{t+l}_{t+1}}, \Pi, Z_t\right)\right) \nonumber \\
&=  \arg \min_{\hat{X}(\cdot):\Acal_Z^{t+l} \rightarrow \Acal_{\hat{X}}} \sum_{x\in\Acal_X} \mathrm{\Lambda}\left(x, \hat{X}\left(Z^{t+l}\right)\right) \frac{\Pi\left(x, Z_t\right) \left(\Pi^{-T}(x,\cdot) Q_{Z_{t}|Z^{t-1},Z^{t+l}_{t+1}}\right)}{ Q_{Z_t|Z^{t-1},Z^{t+l}_{t+1}}\left(Z_t|Z^{t-1},Z^{t+l}_{t+1}\right)}. \label{eq: universal filter with lookahead}
\end{align}

These filters are universal in the sense that, for stationary sources, their normalized expected loss converges to the corresponding loss of the optimal filter.


\begin{definition}[Denoisability]
Let $\mathbf{X}$ and $\mathbf{Z}$ be jointly stationary. For fixed integer $k$, the denoisability is defined as:
\[ \mathbb{D}(\mathbf{X}, \mathbf{Z}, k) \triangleq \lim_{n \to \infty} \Ebb \left[ \frac{1}{n} \sum_{t=1}^n \mathbf{\Lambda} \left( X_t, \hat{X}_t^\text{opt}\left(Z^{t+k}\right) \right) \right].\]
\end{definition}
The denoisability represents the optimum distribution dependent performance of an estimator that estimates the symbols from $\mathbf{X}$ given the symbols from $\mathbf{Z}$, in the limit of large $n$, for fixed lookahead/delay $k$. That the limit in the right-hand side exists is a consequence of the monotonicity of the expression in $n$. $\mathbb{D}(\mathbf{X}, \mathbf{Z}, k)$ can be given a ``single-letter'' characterization analogous to that for denoising in Claim 2 of \cite{weissman2005universal}. The following is a direct consequence of the theorems in Section~\ref{bounding_excess_loss}.
\begin{theorem}
Let SPA $Q$ be universal in the sense of \eqref{eq:universal_SPA} and $\hat{X}^{ao}$ 
be the induced filters, as defined in \eqref{eq: universal filter with delay} and \eqref{eq: universal filter with lookahead}. 
Then, for all stationary $(\mathbf{X}, \mathbf{Z})$ connected via the channel $\Pi$, and every integer $k$,  

    \[\lim_{n \to \infty} \Ebb\left[ \frac{1}{n} \sum_{t=1}^{n}\mathrm{\Lambda}\left(X_t, \hat{X}^{ao}_{t}\left(Z^{t+k}\right)\right) \right] = \mathbb{D} (\mathbf{X}, \mathbf{Z}, k) .\]
\end{theorem}

\section{Implementation and Practical Considerations} \label{implementation}
The universal filters of the previous section are applicable in real world scenarios where exact models are unavailable, especially when based on SPAs with computationally efficient implementations. 
\subsection{Universal Causal Filtering Algorithm}
This subsection outlines the steps required for estimating the channel matrix and filtering a noisy signal. The process is as detailed below:
\begin{itemize}
    \item Input:
    \begin{itemize}
        \item Noisy signal $z^n$
        \item Training data for channel matrix estimation: Historical pairs of $(x_{t},z_{t})$
        \item SPA $Q_{Z_t|z^{t-1}}$
        \item Symbol alphabets $\Acal_{X}$ and $\Acal_{Z}$ of noise-free and noisy signals, respectively
    \end{itemize}
    \item Output:
    \begin{itemize}
         \item Estimate of channel matrix $\Pi$ 
        \item Filtered signal $\hat{X}^{ao}$
    \end{itemize}
    \item The filtering process is as follows:
    \begin{enumerate}
        \item The estimation of $\Pi$
            \begin{itemize}
                \item Initialize a prior distribution model: For $|\Acal_{x}|$ is the size of the alphabet $\Acal_{X}$, start with an initial matrix for $\Pi(X, Z)$ of dimension $|\Acal_{X}| \times |\Acal_{Z}|$ 
                \item Empirical estimation of transition probabilities: Count the occurrences of each pair of $(X,Z)$ in the training data and take\footnote{For notational simplicity, we use $\Pi$ also to denote the estimated rather than the actual channel matrix when the latter is unavailable.}
                \begin{equation}
                    \Pi(x,z) = \hat{P}(Z_t=z|X_t=x) = \frac{\text{Count}(Z=z,X=x)}{\text{Count}(X=x)}.
                \end{equation}
                
            \end{itemize}

        \item Compute filtered estimate, iterating through the signal $Z^n$. For every $t$ between $1$ and $n$, produce the corresponding value of $\hat{X}^{ao}_t$ as follows:: 
            \begin{itemize}
                \item With $Q_{Z_t|z^{t-1}}$ and $\Pi$ calculate the term $ \hat{P}_{X_t|z^t} = F\left(Q_{Z_t|z^{t-1}}, \Pi, z_t\right)$
                \item Loop: For each possible value $x \in \Acal_{X}$, compute the weighted loss using 
                \[\mathrm{\Lambda}\left(x, \hat{X}\left(z^t\right)\right) \times \hat{P}_{X_t|z^t}(x). \]
                \item Select the value of $\hat{X}(z^t)$ that minimizes the summed weighted loss.
            \end{itemize}
    \end{enumerate}
\end{itemize}
\subsection{Causal Filter with Delay}
Although the SPA is specified in terms of  $Q_{Z_t|Z^{t-1}}, \forall t \leq n$, we are required to calculate $Q_{Z_{t}|Z^{t-d}}$ where $d \geq 1$. This calculation boils down to
marginalizing out $Z_{t-d+1}^{t-1}$ from  $Q_{Z_{t-d+1}^t|Z^{t-d}}$ which has complexity growing exponentially in $d$.  
Specifically, if the original SPA had complexity $O(1)$ per time step, adjusting for a delay increases this to approximately $O(|\Acal_Z|^d)$. This complexity can make the computation infeasible for large values of $d$ or $\Acal_Z$.  

Using Monte Carlo simulations to approximate $Q_{Z_t|Z^{t-d}}$ 
can be an effective way to manage this complexity. This approach involves generating a large number of sample sequences $Z_{t-d+1}^t$ from $Q_{Z_{t-d+1}^t|Z^{t-d}}$, 
observing the outcomes of $Z_t$,  
\cite{7339702,1550195}. 
The computational complexity is $O(M d)$ (instead of exponential in $d$), where $M$ is the number of simulations. By controlling $M$, one can balance accuracy and computational complexity.
A higher $M$ will improve the precision, but even modest values can yield reasonable approximations, as is quantified in the following.    
\begin{theorem} \label{thm:montecarlo}
    For $d \geq 1$, let $\hat{Q}_{Z_t|Z^{t-d}}^M$ be an empirical distribution calculated by sampling from $Q_{Z_t|Z^{t-d}}$ $M$ times. 
\begin{align*}
    &\Ebb\left[ \frac{1}{n} \sum_{t=1}^{n}\mathrm{\Lambda}\left(X_t, \hat{X}_B\left(\Pi^{-T}\hat{Q}_{Z_t|Z^{t-d}}^M \right)\right) - \frac{1}{n} \sum_{t=1}^{n} \mathrm{\Lambda}\left(X_t, \hat{X}_t^{\text{opt}}\left(Z^{t-d}\right)\right)\right] \\
    &\leq \mathrm{\Lambda}_\text{max} \left( C_{2}(\Pi) \sqrt{\frac{d}{n}D\left(P_{Z^n}\Vert Q_{Z^n}\right)} + C_{3}(\Pi) \frac{1}{\sqrt{M}} \right),
\end{align*}
where 
\[C_2(\Pi) = 3 \sqrt{2} \left|\Pi^{-T}\right|_\text{max} \left| \Acal_X \right| \]
\[C_3(\Pi) = \left|\Pi^{-T}\right|_\text{max} \left| \Acal_X \right| 2^{|\Acal_Z|} \sqrt{\frac{\pi}{2}} \]
\end{theorem}
 First, on LHS, add and subtract the normalized loss term for the case using $\hat{X}_B\left(\Pi^{-T}Q_{Z_t|Z^{t-d}} \right)$. Now we have four terms, and then we use Theorem~\ref{thm:normalizedloss_delay} to bound the difference between two of those terms. For the difference between the other two, we use Lemma~\ref{lem:loss_bound_arbitrary} to first bound it by the sum of two L1 distance terms. Then use Theorem~\ref{thm:normalizedloss_delay} and  Lemma~\ref{lem:empirical_distr_bound}.

\subsection{Non-causal Filter with Lookahead}

In the case of non-causal filtering with lookahead $l \geq 0$, given $Q_{Z_t|Z^{t-1}}$ from SPA, using Bayesian rule we can calculate $Q_{Z_t|Z^{t-1},Z^{t+l}_{t+1}}$ as follows:
\begin{equation}
Q\left(Z_t|Z^{t-1},Z^{t+l}_{t+1}\right) = \frac{Q\left(Z_t|Z^{t-1}\right)\cdot Q\left(Z^{t+l}_{t+1}|Z^{t-1},Z_t\right)}{\sum_{a\in \Acal_Z}Q\left(a|Z^{t-1}\right)\cdot Q\left(Z^{t+l}_{t+1}|\left(Z^{t-1},a\right)\right)}, 
\end{equation}
where
\begin{equation}
Q\left(Z^{t+l}_{t+1}|Z^{t-1},Z_t\right) = \prod_{k=t+1}^{t+l} Q\left(Z_k|Z^{k-1}\right).    
\end{equation}
This approach efficiently incorporates future observations $Z_{t+1}^{t+l}$ into the estimation process by leveraging the chain rule of probability. Notably, the complexity of this computation per time step is linear in $l$, the length of the lookahead window, which is practical for moderate $l$.

\section{Experiments} \label{experiment}
In this section, we present our experiments on implementing a universal filter induced by a universal SPA. Although universal SPAs can be accessed through a variety of established methods—including Context Tree Weighting (CTW)\cite{382012} and Prediction by Partial Matching (PPM)\cite{1096090}—each with its strengths and limitations in capturing dependencies, estimating probabilities, redundancy rates and complexity, we focus here specifically on the celebrated LZ78-based SPA. This choice enables efficient and adaptable handling of diverse data structures by leveraging the LZ78 algorithm's ability to capture increasingly elaborate context information with complexity essentially linear in the amount of data processed. We leave experimentation with filters induced by other universal SPAs to future work. 



To assess the performance of the universal filter, we apply it to a Markov process setting,  allowing us to benchmark against established linear filtering techniques, as well as the optimum Bayes performance.
 In particular, we compare our method’s performance with the Wiener filter, a traditional linear filtering technique, as well as with the theoretical optimum of non-linear filtering for this setting. 

\subsection{LZ78 Compression Algorithm}
In this work, we use LZ78 to create an SPA $Q_{Z_t|z^{t-1}}$, which is an estimate of $\Pbb\left(Z_t=z|Z^{t-1}=z^{t-1}\right)$. LZ78 encoding, as described in \cite{originalLZ78paper}, is based on constructing a prefix tree. This tree is created by dividing a given sequence into a series of consecutive subsequences, referred to as phrases.

\begin{definition}
[LZ78 Tree]
    The LZ78 tree constructs phrases from a sequence through the following procedure: 
    \begin{enumerate} 
        \item The LZ78 tree begins with a single root node. 
        \item Repeat the following until reaching the end of the sequence: 
        \begin{enumerate} 
            \item Begin at the root and traverse down the tree according to the upcoming symbols in the sequence until reaching a leaf node.
            \item Add a new node extending from this leaf, representing the next symbol in the sequence.
            \item Each LZ78 phrase is formed from the sequence portion used in the traversal, combined with the symbol of the new branch just created. 
    
    \end{enumerate} 
\end{enumerate}
\end{definition}
The number of nodes in the prefix tree, excluding the root, is equivalent to the number of phrases that have been parsed.


\subsection{Defining the LZ78 Tree of Sequential Probability Assignments}

To address the limitations of naive empirical frequency-based SPAs, we leverage the LZ78 tree structure to create an adaptive SPA.  The naive SPA, as defined in \cite{sagan2024familylz78baseduniversalsequential}: 
\[q^\text{naive}\left(a|Z^{t-1}\right) = \frac{N_{Z_t}^{LZ}\left(a|Z^{t-1}\right)}{t-1}.\]
$N^{LZ}_{Z_t}\left(a|Z^{t-1}\right)$, $a \in \Acal_Z$, is the number of times of this symbol being a prefix of phrases up to $Z^{t-1}$ in $\mathcal{L}\left(Z^{t-1}\right)$, where $a$ is a potential ``next symbol'' in alphabet $\Acal_Z$ of sequence at position $Z_t$ . However, this naive SPA suffers from infinite log loss \footnote{The quality of an SPA is usually evaluated via log loss.} in cases where a symbol has not appeared before. 

To mitigate this, we use the Dirichlet SPA defined in \cite{sagan2024familylz78baseduniversalsequential}: 
\[q^{LZ78}\left(a|Z^{t-1}\right) = \frac{N^{LZ}_{Z_t}\left(a|Z^{t-1}\right) + \gamma}{\sum_{b \in \Acal_Z} N^{LZ}_{Zt}\left(b | Z^{t-1}\right) + \gamma|\Acal_Z|}.\]
 The choice of $\gamma = \frac{1}{2}$ is essentially (to the leading term) minimax optimal with respect to the log loss incurred by the SPA on any individual sequence \cite{sagan2024familylz78baseduniversalsequential}.

\subsection{Optimality of the LZ78 Family of SPAs}
The LZ78 family of SPAs are optimal in the sense that, for any individual sequence, they achieve a limit supremum on log loss that matches or surpasses the performance of any finite-state SPA \cite{sagan2024familylz78baseduniversalsequential}. This universality result demonstrates their effectiveness in both individual sequence scenarios and stationary stochastic processes, where the expected log loss approaches the entropy rate of the source as the sequence length increases. Furthermore, for specific classes of sequences, LZ78 SPAs strictly outperform finite-state SPAs, underscoring their superiority in various prediction tasks. This optimality forms the basis for choosing LZ78 SPAs in our experiments.
\subsection{Modified LZ78 Tree Algorithm}
In this section, we explain the modified universal filtering algorithm with LZ78 SPA.  

\subsubsection{Input Shifting}
To improve the efficiency of the LZ78 SPA tree given limited training data, we propose a modified training approach that involves input shifting. Instead of expanding the tree strictly based on each parsed phrase, we grow it by advancing the input by one symbol at a time and starting each traversal at the root. Specifically, after processing a sequence, the next training sequence begins with the input shifted by one symbol relative to the previous sequence. This approach ensures that each symbol sequence in the dataset contributes to the tree from multiple starting points.

By implementing this shifting mechanism, each sequence in the dataset has the opportunity to traverse the tree multiple times, starting from different positions. This repetition effectively increases the training frequency of symbol sequences, even under conditions of limited data. The result is a more robust and well-trained LZ78 SPA tree structure that better captures the probabilistic relationships inherent in the dataset.

The underlying assumption of this method is that the noise characteristics in the dataset are stationary, meaning they remain consistent across shifts. Under these conditions, the probability assignments for sequences are preserved across the multiple traversals introduced by shifting, enhancing the model's overall predictive capabilities. Such an approach is also seen in \cite{Begleiter_2004}.
\subsubsection{LZ78 Tree Pruning}
In the context of LZ78-based SPA, pruning infrequently traversed branches in the probability assignment tree effectively reduces model complexity while preserving predictive accuracy. By introducing a threshold parameter $N_{th}$, branches that do not meet a minimum traversal count are pruned, effectively excluding low-occurrence events that contribute less significantly to the overall probability estimation. This approach prevents overfitting to rare events and helps maintain an SPA that remains close to the true probability distribution. Moreover, pruning enhances computational efficiency by focusing on well-traversed branches, particularly in resource-constrained environments or applications demanding rapid sequential probability access. This methodology is crucial for practical implementations, allowing the SPA to achieve a balanced trade-off between prediction fidelity and computational tractability.

\begin{algorithm}[h]
\caption{Universal Algorithm}\label{alg:filtering}
\begin{algorithmic}[1]
\STATE Initialization $\Pi \gets [0]$, $\Pi \in \mathbb{R}^{|\Acal_X| \times |\Acal_Z|}$; \COMMENT{Channel Matrix}
\STATE Output $\hat{X}^{ao} \gets [\;]$\COMMENT{Filtered output}
\STATE Count the occurrences of each historical pair of (X,Z) in the training data.
\STATE Compute: $\Pi^{approx}(x,z) = \hat{P}(Z_t=z|X_t=x) = \frac{\text{Count}(Z=z,X=x)}{\text{Count}(X=x)}$
\STATE Generate LZ78 SPA Tree from training data
\IF{Causal Universal Filter}
    \item \[Q(z) = Q_{Z_t|Z^{t-1}}(z) = \frac{N^{LZ}_{Z_t}\left(z|Z^{t-1}\right) + \gamma}{\sum_{b \in \Acal_Z} N^{LZ}_{Z_t}\left(b | Z^{t-1}\right) + \gamma |\Acal_Z|}\]
\ELSIF{Causal Universal Filter with Delay $d \geq 1$}  
    \item \[Q(z) = Q_{Z_t|Z^{t-d}}(z) = \frac{\sum_{\forall Z_t = z} N^{LZ}_{Z_t}\left(z|Z^{t-d}\right) + \gamma}{\sum_{b \in \Acal_Z} N^{LZ}_{Z_t}\left(b | Z^{t-d}\right) + \gamma |\Acal_Z|} \] 
    \COMMENT{$\sum_{\forall Z_t = z} N^{LZ}_{Z_t}\left(z|Z^{t-d}\right)$
    marginalizes all branches with $Z_t$th order symbol being z}
\ELSIF{Non-causal Universal Filter with Lookahead $l \geq 0$}
    \item \[Q(z) = Q_{Z_t|Z^{t-1},Z^{t+l}_{t+1}}(z) = \frac{ N^{LZ}_{Z_{t+l}}\left(Z_{t+l}|Z^{t-1}, z, Z_{t+1}^{t+l-1}\right) + \gamma}{\sum_{b \in \Acal_Z} N^{LZ}_{Z_{t+l}}\left(Z_{t+l}|Z^{t-1}, b, Z_{t+1}^{t+l-1}\right) + \gamma |\Acal_Z|}\]
\ENDIF
\WHILE{$t \leq n$}
    \STATE With $Q$ and $\Pi$ calculate the term $ \hat{P}^{approx} = F\left(Q, \Pi^{approx}, Z_t\right)$
    \STATE Loop: For each possible value $x \in \Acal_{X}$, compute the weighted loss using 
                \[\mathrm{\Lambda}\left(x, \hat{X}\left(Z^t\right)\right) \times \hat{P}^{approx}(x) \]
    \STATE Output $\gets$ the value of $\hat{X}\left(Z^t\right)$ that minimizes the summed weighted loss. 
\ENDWHILE
\end{algorithmic}
\end{algorithm}

\begin{algorithm}[h]
\caption{Universal Algorithm with Delay using Monte Carlo}\label{alg:filtering_MC}
\begin{algorithmic}[1]
\STATE Initialization $\Pi \gets [0]$, $\Pi \in \mathbb{R}^{|\Acal_X| \times |\Acal_Z|}$; \COMMENT{Channel Matrix}
\STATE Output $\hat{X}^{ao} \gets [\;]$\COMMENT{Filtered output}
\STATE Count the occurrences of each historical pair of (X,Z) in the training data.
\STATE Compute: $\Pi^{approx}(x,z) = \hat{P}(Z_t=z|X_t=x) = \frac{\text{Count}(Z=z,X=x)}{\text{Count}(X=x)}$
\STATE Monte Carlo: For each experiment, start at the root of the LZ78 SPA tree, traverse the tree with sequence $Z^{t-d}$, and generate the next D symbols by sampling from the probability distribution defined at each tree node. 
\STATE Count the occurrence of $z\in \Acal_Z$ in the last symbol of generated sequences. 
    \item \[Q = Q_{Z_t|Z^{t-d}}\left(Z_t = z|Z^{t-d}\right) = \frac{\sum_{\forall Z_t = z} N^{MC}_{Z_t}\left(z|Z^{t-d}\right) + \gamma}{\sum_{b \in \Acal_Z} N^{MC}_{Z_t}\left(b | Z^{t-d}\right) + \gamma |\Acal_Z|} \] 

\WHILE{$t \leq n$}
    \STATE With $Q$ and $\Pi$ calculate the term $ \hat{P}^{approx} = F\left(Q, \Pi^{approx}, Z_t\right)$
    \STATE Loop: For each possible value $x \in \Acal_{X}$, compute the weighted loss using 
                \[\mathrm{\Lambda}\left(x, \hat{X}\left(Z^t\right)\right) \times \hat{P}^{approx}(x) \]
    \STATE Output $\gets$ the value of $\hat{X}\left(Z^t\right)$ that minimizes the summed weighted loss. 
\ENDWHILE
\end{algorithmic}
\end{algorithm}

\subsection{Markov case experiment setup} \label{Markov_setup}
To check our LZ78 tree algorithm's performance, we conducted an experiment on simulated data where the noise-free data is a Markov process and it's corrupted by additive noise. 

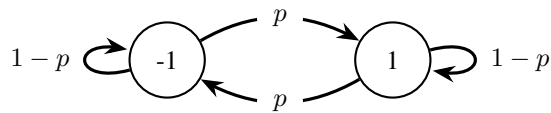
\begin{figure}[h]
    \centering

    \begin{tikzpicture}
    \begin{scope}[every node/.style={circle,thick,draw, minimum size = 1cm}]
        \node (-1) at (0,0) {-1};
        \node (1)  at (3,0) { 1};
    \end{scope}
    \begin{scope}[>={Stealth[black]},
                every node/.style={fill=white,circle},
                every edge/.style={draw=black,very thick}]

        \path [->] (-1) edge[bend left=30] node {$p$} ( 1); 
        \path [->] ( 1) edge[bend left=30] node {$p$} (-1); 
        \path [->] ( 1) edge[loop right] node {$1-p$} ( 1); 
        \path [->] (-1) edge[loop left] node {$1-p$} (-1); 
    \end{scope}

    \end{tikzpicture} 
    
    \caption{Two-state Markov Source}
    \label{fig:markov_source}
\end{figure}

The noise-free data $X_t$ is generated from a first-order symmetric binary Markov source with transition probability $p$. It is corrupted by an iid additive noise $N_t$, distributed as:
\begin{equation}
    N_{t} =
    \begin{cases}
      -1 & \text{w.p. $\frac{1}{2}$}\\
      1 & \text{w.p. $\frac{1}{2}$}
    \end{cases}       
\end{equation}
\begin{equation}
    Z_{t} =  X_{t}+N_{t}
\end{equation}

For this noisy channel, the channel matrix $\Pi$ can be taken as:
\begin{equation}
    \Pi = 
    \begin{bmatrix}
        \frac{1}{2} & \frac{1}{2} & 0 \\
        \frac{1}{3} & \frac{1}{3} & \frac{1}{3} \\
        0 & \frac{1}{2} & \frac{1}{2}
    \end{bmatrix}
\end{equation}
where $X_{t}\in \{ -1,0,1 \}$ and $Z_{t}\in \{-2,0,2 \}.$ Note that $X_t$ can only be either -1 or 1, but to comply with the setting of equal alphabets for both $X_t$ and $Z_t$, we added 0 to $X_t$'s alphabet. Since $\Pr(X_t=0)=0$, the second row of $\Pi$ is in fact a degree of freedom, and any probability vector is a valid choice\footnote{An alternative to this approach would be to leave $\Pi$ as a $2 X 3$ matrix and use its Moore–Penrose generalized inverse, as suggested in \cite{weissman2005universal}}.

We generate $X_1$ by sampling from -1 or 1 with equal probability and the remaining components of $X^n$ according to the  Markov chain described above. We then corrupted by the additive noise to produce  $Z^n$. For filters requiring pretraining, subsequences of length 9,920,000 are extracted from the noise-free and noisy sequences for training. Performance evaluation is conducted using test subsequences of length 80,000 extracted from each sequence. The mean squared error (MSE) between the filtered subsequences and their corresponding noise-free subsequences is computed to quantify performance.

\subsection{Baseline 1: Theoretical Limit}

In this setting, we evaluate the theoretical limit of expected loss by deploying the Bayes optimal schemes, using forward and backward recursions \cite{ephraim2002hidden}.
Let $a_{X_{t-1}, X_t}$ denote the transition probability from $X_{t-1}$ to $X_t$, and $b\left(Z_t|X_t\right)$ be an probability of observing $Z_t$ given the state $X_t$. To better distinguish $P\left(X_t|Z^t\right)$ from $P\left(X_t|Z^{t-1}\right)$, let $\alpha\left(X_t|Z^t\right) = P\left(X_t|Z^t\right)$.\\
Forward recursion is defined as follows.
\[\alpha\left(X_t|Z^t\right) = \frac{P\left(X_t|Z^{t-1}\right) b\left(Z_t|X_t\right)}{\sum_{X_t=1}^M P\left(X_t|Z^{t-1}\right) b\left(Z_t|X_t\right)}, \qquad t=1, 2, \cdots, n\]
\[P\left(X_t|Z^{t-1}\right) = \sum_{X_{t-1}=1}^M a_{X_{t-1}, X_t} \alpha\left(X_{t-1}|Z^{t-1}\right), \qquad t=2, 3, \cdots, n\]
where $P\left(X_1|Z^0\right)$ is an initial distribution of $X_1$ and $M$ is a size of alphabet of $X_t$. 

Using forward recursion, we can get the theoretical limit of causal filter with or without delay. First, for the causal filter without delay, we can recursively calculate $\alpha\left(X_t|Z^t\right)=P\left(X_t|Z^t\right)$ starting from $P\left(X_1|Z^0\right)$ since we know $a_{X_{t-1}, X_t}$ and $b\left(Z_t|X_t\right)$. Next, for the causal filter with delay $d \geq 1$, we can calculate $\alpha\left(X_{t-d}|Z^{t-d}\right)=P\left(X_{t-d}|Z^{t-d}\right)$ similarly. Then we can get $P\left(X_t|Z^{t-d}\right)$ via
\begin{align*}
    P\left(X_{t-d+k} | Z^{t-d}\right) &= \sum_{X_{t-d+k-1}} P\left(X_{t-d+k}, X_{t-d+k-1} | Z^{t-d}\right) \\
    &= \sum_{X_{t-d+k-1}} P\left(X_{t-d+k} | X_{t-d+k-1} Z^{t-d}\right) P\left(X_{t-d+k-1} | Z^{t-d}\right) \\
    &= \sum_{X_{t-d+k-1}} a_{X_{t-d+k-1}, X_{t-d+k}} P\left(X_{t-d+k-1} | Z^{t-d}\right),  \qquad \qquad k=1, 2, \cdots, d
\end{align*}
The backward recursion can be dervied as 
\[P\left(X_t|Z^n\right) = \alpha\left(X_t|Z^t\right) \sum_{X_{t+1}=1}^M \frac{a_{X_t, X_{t+1}} P\left(X_{t+1}|Z^n\right)}{P\left(X_{t+1}|Z^t\right)}, \qquad t=n-1, n-2, \cdots, 1\]
where $P\left(X_n|Z^n\right) = \alpha\left(X_n|Z^n\right)$.
We can calculate $\alpha\left(X_t|Z^t\right), \alpha\left(X_{t+1}|Z^{t+1}\right), \cdots, \alpha\left(X_{t+l}|Z^{t+l}\right)$ and \\
$P\left(X_{t+1}|Z^t\right), P\left(X_{t+2}|Z^{t+1}\right), \cdots, P\left(X_{t+l}|Z^{t+l-1}\right)$ by forward recursion, then we can get $P\left(X_t|Z^{t+l}\right)$ by the backward recursion.

\subsection{Baseline 2: Wiener Filter}
The Wiener filter is a linear FIR filter designed to minimize the mean squared error between a desired signal and its estimate. It is widely used in signal processing for noise reduction due to its simplicity and optimality in linear Gaussian environments. In our study, we benchmarked our universal filter against the Wiener filter, ensuring comparable computational resources for both. We fixed the filter length (also referred to as window size in some parts of the literature) to correspond to the phrase length considered for the LZ78-based universal filter. Additionally, we evaluated the Wiener filter in scenarios with a delay and in a non-causal setup with lookahead. 
\subsection{Results}
    \begin{figure}
    \centering
    \includegraphics[width=0.5\textwidth]{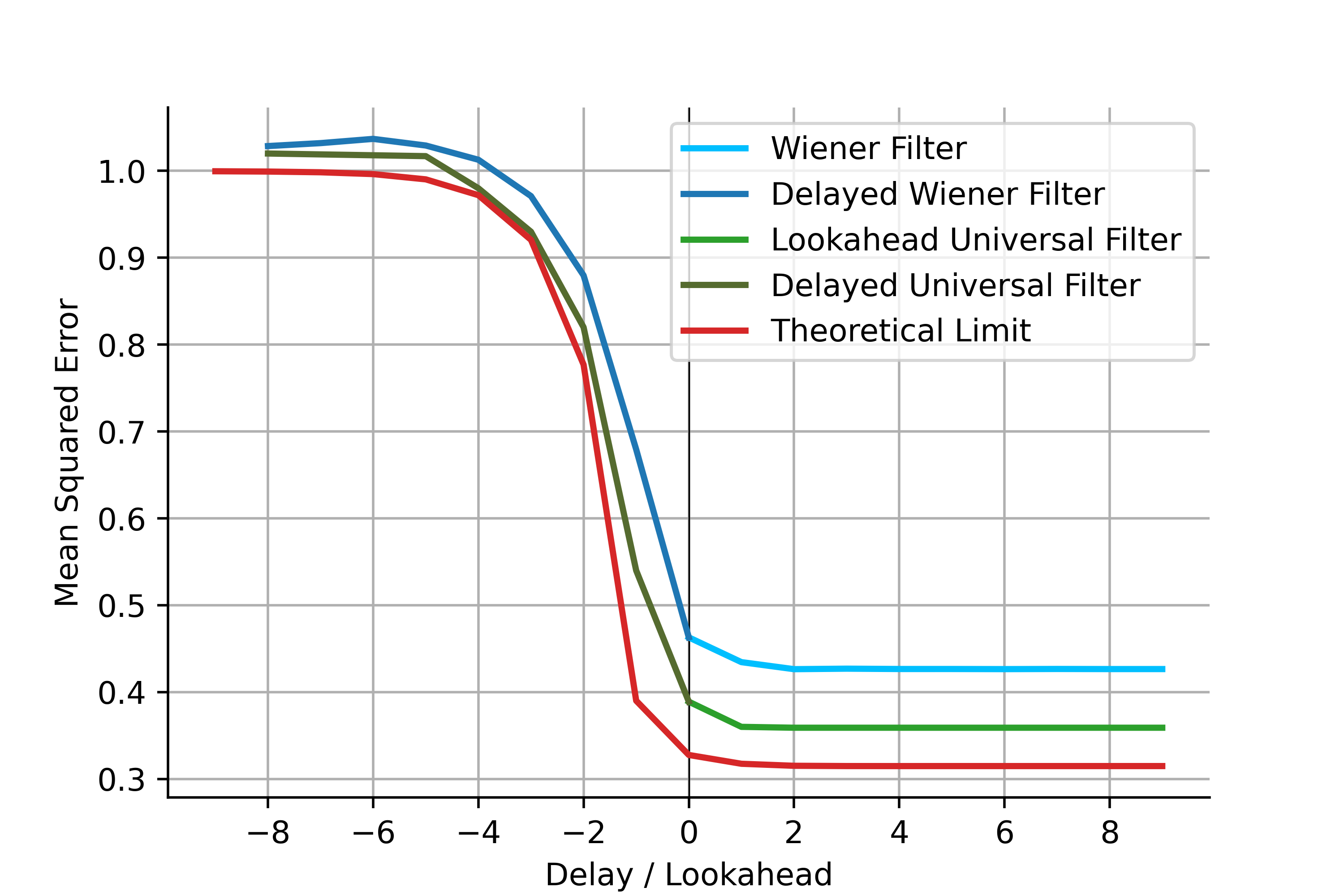}
    \caption{Comparing the MSE loss of universal filter to a linear Wiener filter and the theoretical limit}
    \label{fig:losses_benchmark}

    \end{figure}
The plot in Figure {\ref{fig:losses_benchmark}} compares the Mean Squared Error (MSE) performance of the various filters across different delay/lookahead values, benchmarked against the theoretical MSE limit for this controlled Markov process. The x-axis represents the delay/lookahead values, with negative values indicating delay and positive values indicating lookahead. The y-axis displays the resulting MSE. Using our proposed LZ78-based  universal filter demonstrates a substantial reduction in MSE, getting much closer to the theoretical limit. 

Next, we consider the performance and computation trade-off of the universal filter in the case of delay, using the Monte Carlo method discussed in lieu of precise computation of $Q_{Z_t|Z^{t-d}}$,  experimenting across different configurations, incorporating 100, 1000, and 10,000 Monte-Carlo trials. The results are summarized in Figure {\ref{fig:losses_benchmark_MC}}. As may be expected, at smaller Monte Carlo sample sizes,  
the performance of the filter slightly lags behind that based precisely on $Q_{Z_t|Z^{t-d}}$.
 However, as the sample size increases, the performance becomes indistinguishable from that of the latter.

\begin{figure}
    \centering
    \includegraphics[width=0.5\textwidth]{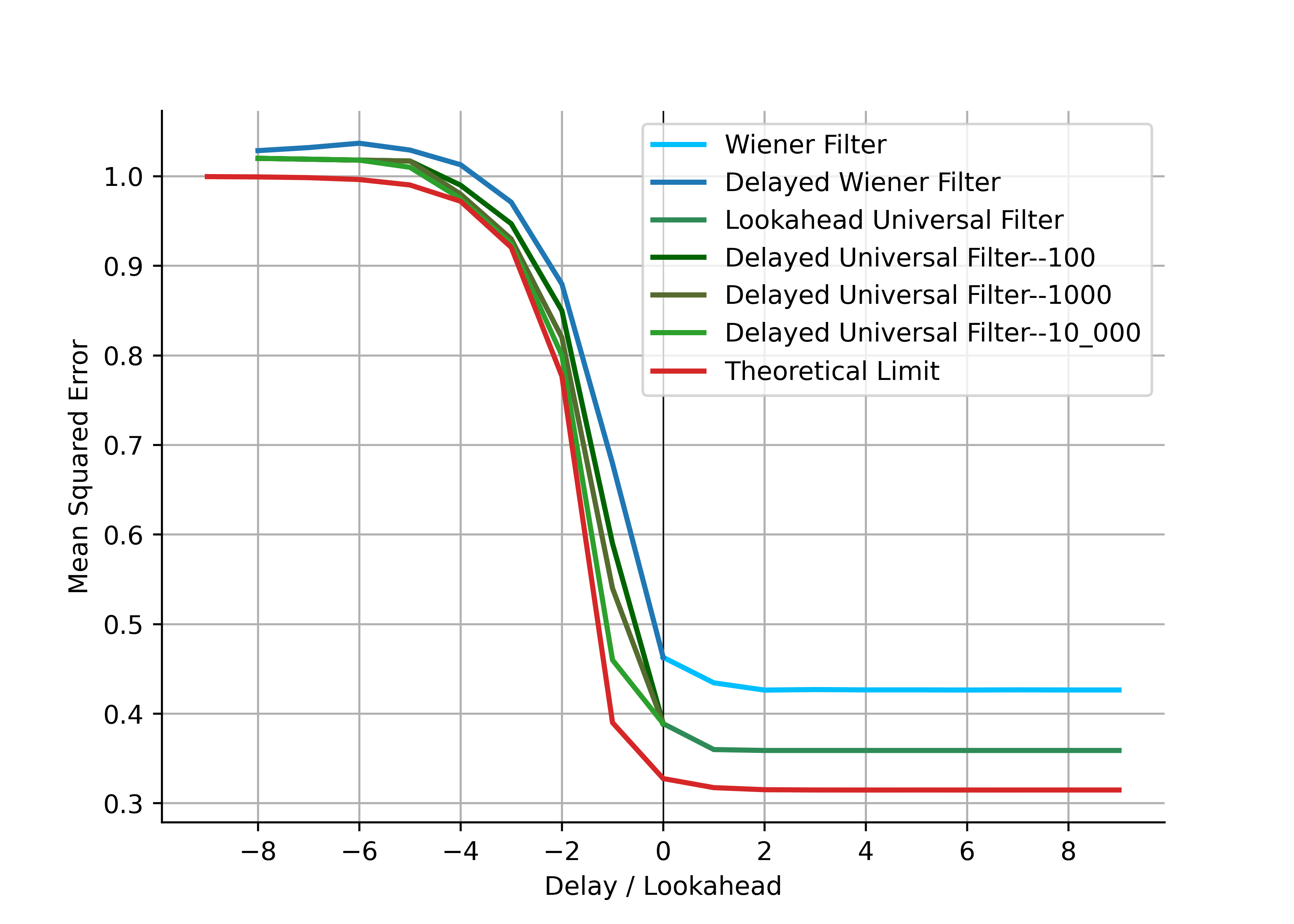}
    \caption{MSE loss of universal filter, increasing the number of Monte Carlo steps}
    \label{fig:losses_benchmark_MC}

\end{figure}

\subsection{Discussion}
One of the central considerations for real-time applications in low-compute environments is the computational complexity of the employed model. Our proposed approach, based on the LZ78 SPA, demonstrates significant computational advantages over neural network-based methods. This efficiency makes it particularly well-suited for environments where resources are limited and latency is critical. 

The efficiency of the LZ78-based predictor stems from its ability to compute the conditional probability distribution $Q_{Z_t|Z^{t-1}}$ with $O(1)$ complexity. While this is achieved through the inherent structure of the LZ78 tree, it is not restricted to this specific implementation.  The filtered output relies on simple arithmetic operations (multiplications and divisions) and a summation over the finite set $\Acal_X$, the alphabet of possible symbols. As a result, the overall computational cost scales linearly with the size of the alphabet, $|\Acal_X|$, ensuring this approach remains efficient for real-time tasks.

In contrast, neural network-based predictors for the same denoising filter task involve significantly higher computation overhead. For filtering problems, neural network methods require at least $Q\left(n^2\right)$ complexity\footnote{For the simplest form of a neural network, single-layer perceptron, in filtering problem, we have n inputs and n output neurons, the forward pass is of $O\left(n^2\right)$ complexity. }.
Additionally, neural networks typically require substantial memory resources to store weights, activations, and intermediate computations. These resource demands often necessitate specialized hardware, such as GPUs, to meet real-time constraints. This reliance on computationally expensive resources makes neural network-based methods less practical for low-compute environments, particularly when simplicity and efficiency are critical.

\subsection{Lower Bound for Markov Case}

    \begin{figure}
    \centering
    \includegraphics[width=0.5\textwidth]{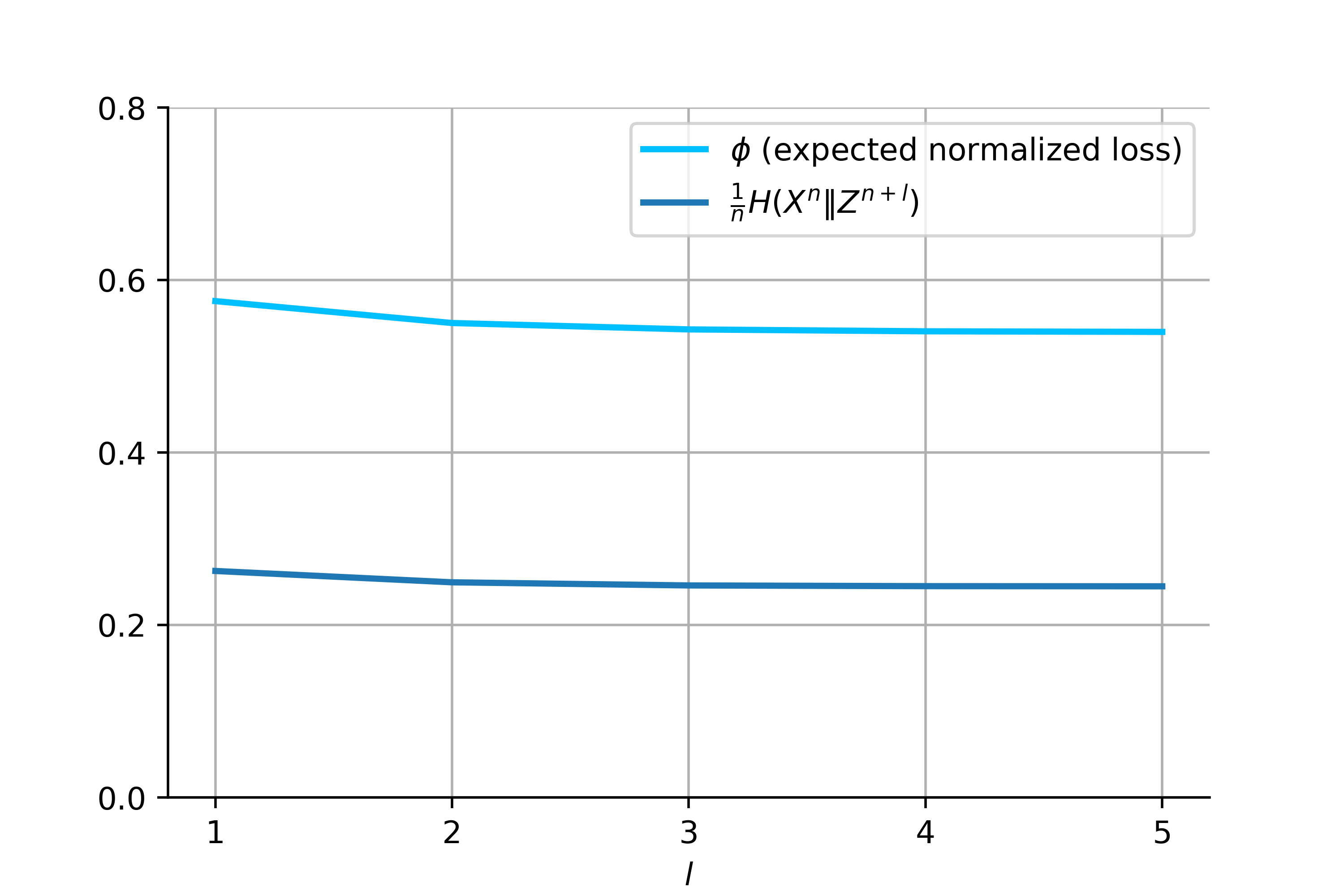}
    \caption{Numerical results based on Theorem~\ref{thm:mutual-information-lower-bounds} where $n=5$}
    \label{fig:lowerbound_lookahead}

    \end{figure}

    \begin{figure}
    \centering
    \includegraphics[width=0.5\textwidth]{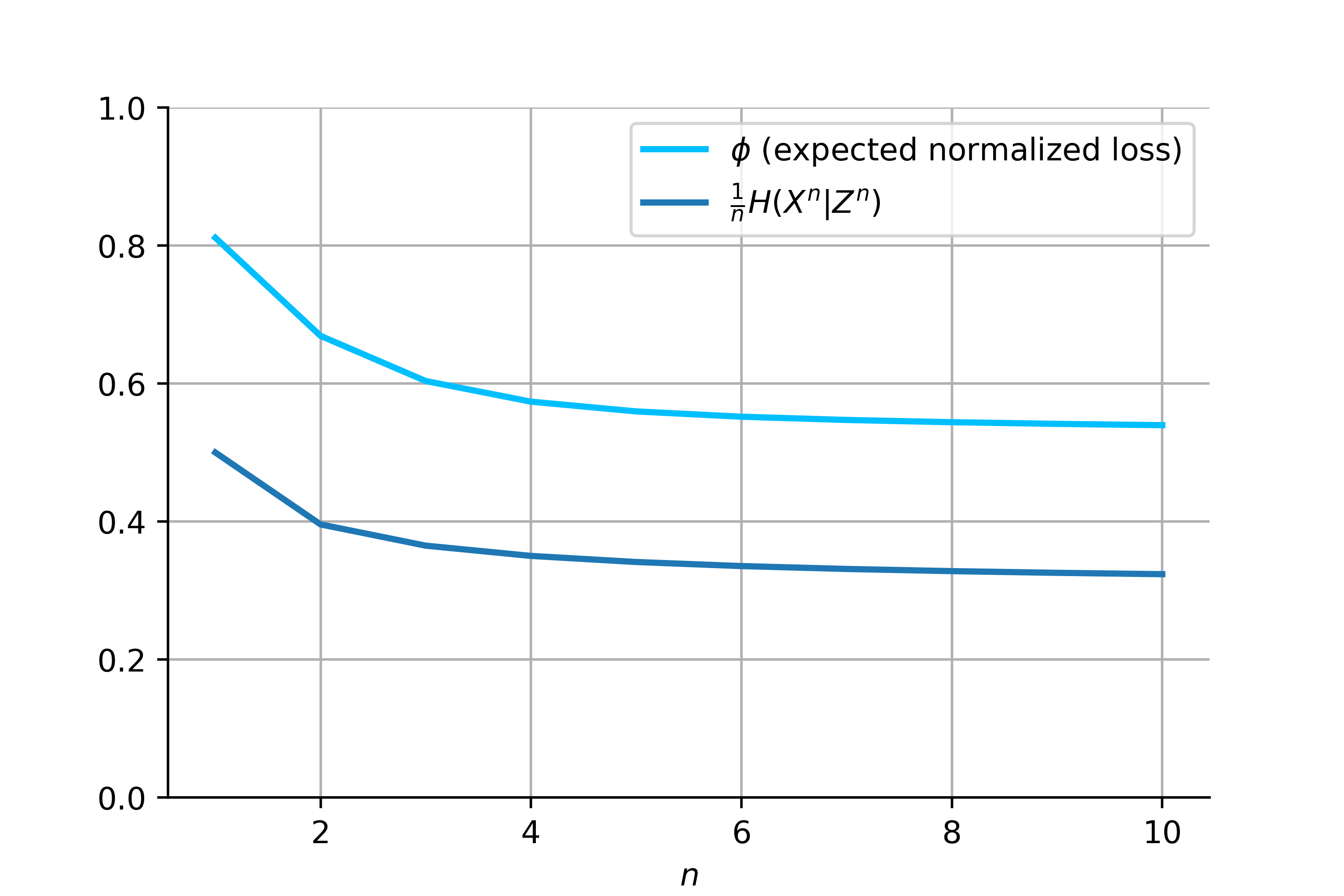}
    \caption{Numerical results based on Corollary~\ref{cor:mutual-information-lower-bounds}}
    \label{fig:lowerbound_noncausal}

    \end{figure}

In this section, we show how the lower bounds of Theorem~\ref{thm:mutual-information-lower-bounds} and Corollary~\ref{cor:mutual-information-lower-bounds} manifest in the setting (noise-free source and channel) of Section \ref{Markov_setup}, under the Hamming loss function. In this case, $\phi(D)$ can be represented as:
\begin{align*}
    \phi(D) &= \max \left\{ H(U) : U \text{ random variable over } \mathcal{A}_X \text{ s.t. } \mathbb{E}[\rho(U)] \leq D \right\} \\
    &= \max \left\{ H(U) : U \text{ random variable over } \mathcal{A}_X \text{ s.t. } P(U=1) \leq D \right\} \\
    &= \begin{cases} 
    -D \log D - (1-D) \log (1-D) & \text{if } D \leq \frac{1}{2}, \\
    1 & \text{otherwise.}
    \end{cases}
\end{align*}

Furthermore, in our setting, it is possible to derive explicit expressions for the expected normalized loss of optimal filters and the information-theoretic lower bounds. We calculated these values and plotted them in Figure \ref{fig:lowerbound_lookahead} and Figure \ref{fig:lowerbound_noncausal}. Figure \ref{fig:lowerbound_lookahead} illustrates the numerically calculated values of $\phi \left( \mathbb{E}\left[ \frac{1}{n} \sum_{t=1}^n \mathrm{\Lambda}\left(X_t,  \hat{X}^{opt}\left(Z^{t+l}\right)\right)\right] \right)$ and $\tfrac{1}{n} H(X^n\lVert Z^{n+l})$ for a fixed sequence length of $n=5$ and varying lookahead values $l$ from 1 to 5. Figure \ref{fig:lowerbound_noncausal} presents those of 
$\phi \left( \mathbb{E}\left[ \frac{1}{n} \sum_{t=1}^n \mathrm{\Lambda}\left(X_t,  \hat{X}^{opt}\left(Z^{n}\right)\right)\right] \right)$ and $\tfrac{1}{n} H(X^n | Z^{n})$ for varying sequence lengths $n$ from 1 to 10. 

\section{Conclusion} \label{sec: conclusions} 
We proposed an approach to the construction of universal filtering schemes that can be used under arbitrary delay or lookahead. Given any universal SPA, we showed that the resulting schemes are universally optimal, with performance converging to that of the Bayes optimum as the length of the sequence grows. Furthermore, we implemented and experimented on simulated data with a family of such schemes when the driving universal SPA is induced by LZ78. The experimental results are consistent with the theoretical ones and bode well for applying such filters on real data.

Our focus here was on the stochastic setting, deferring the semi-stochastic setting, where the noise-free data comprises individual sequences, to future work.  
Additionally, future work will explore broader settings where the noise is neither discrete
\cite{sivaramakrishnan2008universal, sivaramakrishnan2009context, dembo2005universal} 
nor memoryless \cite{zhang2005discrete}, such as in neural recording data. 
Finally, future work will also systematically compare our LZ78-based schemes with neural network-based methods, considering the filtering performance alongside the computational and memory requirements. 


{\appendices
\section{Proofs of filtering with Sequential Probability Assignments }\label{app:SPA}

\textbf{Proof of Lemma~\ref{conditionalcausal}.}
For all $t \leq n$,
    \[P_{X_t|Z^t=z^t} = F\left(P_{Z_t|z^{t-1}}, \Pi, z_t\right),\]
and
\[P_{X_t|Z^n=z^n} = F\left(P_{Z_t|z^{t-1}, z_{t+1}^n}, \Pi, z_t\right)\]
\begin{proof}
    By Bayes' theorem,
    \begin{align*}
        \Pbb\left(X_t=x|Z^t=z^t\right) &= \frac{\Pbb\left(X_t=x, Z^t=z^t\right)}{\Pbb\left(Z^t=z^t\right)} = \frac{\Pbb\left(Z_t=z_t|X_t=x,Z^{t-1}=z^{t-1}\right) \Pbb\left(X_t=x|Z^{t-1}=z^{t-1}\right) \Pbb\left(Z^{t-1}=z^{t-1}\right)}{\Pbb\left(Z_t=z_t|Z^{t-1}=z^{t-1}\right) \Pbb\left(Z^{t-1}=z^{t-1}\right)} \\
        &= \frac{\Pbb\left(Z_t=z_t|X_t=x,Z^{t-1}=z^{t-1}\right) \Pbb\left(X_t=x|Z^{t-1}=z^{t-1}\right)}{\Pbb\left(Z_t=z_t|Z^{t-1}=z^{t-1}\right)}.
    \end{align*}
    As the system mapping $X$ to $Z$ is a discrete memoryless channel, 
    \[\Pbb\left(Z_t=z_t|X_t=x,Z^{t-1}=z^{t-1}\right) = \Pbb\left(Z_t=z_t|X_t=x\right) = \Pi\left(x, z_t\right),\]
    \ie, as we're already conditioning on $X_t$, $Z^{t-1}$ provides no extra information about $Z_t$.

    So, $\Pbb\left(X_t=x|Z^t=z^t\right)$ simplifies to
    \[\Pbb\left(X_t=x|Z^t=z^t\right) = \frac{\Pi\left(x, z_t\right)\Pbb\left(X_t=x|Z^{t-1}=z^{t-1}\right)}{\Pbb\left(Z_t=z_t|Z^{t-1}=z^{t-1}\right)}.\]
    In matrix-vector form,
    \[\Pbb_{X_t|Z^t}\left(\cdot|Z^t=z^t\right) = \frac{\Pi\left(\cdot, z_t\right) \odot \Pbb_{X_t|Z^{t-1}}\left(\cdot|Z^{t-1}=z^{t-1}\right)}{\Pbb\left(Z_t=z_t|Z^{t-1}=z^{t-1}\right)}.\]

    We can evaluate $\Pbb_{X_t|Z^{t-1}}\left(\cdot|Z^{t-1}=z^{t-1}\right)$ using the same process as $\Pbb_{X_t|Z_t}\left(\cdot|Z_t=z_t\right)$.
    By the law of total probability,
    \[\Pbb\left(Z_t=z|Z^{t-1}=z^{t-1}\right) = \sum_{x \in \Acal} \Pbb\left(Z_t=z|X_t=x\right) \Pbb\left(X_t=x|Z^{t-1}=z^{t-1}\right),\]
    so, solving for $\Pbb_{X_t|Z^{t-1}}\left(\cdot|Z^{t-1}=z^{t-1}\right)$ in matrix-vector form,
    \[\Pbb_{X_t|Z^{t-1}}\left(\cdot|Z^{t-1}=z^{t-1}\right) = \Pi^{-\top} P_{Z_t|z^{t-1}}.\]
    Putting it all together,
    \[\Pbb_{X_t|Z^t}\left(\cdot|Z^t=z^t\right) = \frac{\Pi\left(\cdot, z_t\right) \odot \left(\Pi^{-\top} P_{Z_t|z^{t-1}}\right)}{\Pbb\left(Z_t=z_t|Z^{t-1}=z^{t-1}\right)} = F\left(P_{Z_t|z^{t-1}}, \Pi, z_t\right).\]
    Similarly, we can show the result for $P_{X_t|Z^n=z^n}$ by plugging in $P_{Z_t|Z^{t-1}=z^{t-1}, Z_{t+1}^n=z_{t+1}^n}$ instead of $P_{Z_t|Z^{t-1}=z^{t-1}}$.
\end{proof}

\textbf{Proof of Lemma~\ref{lem:optimal_estimator}.} 
Let $W \in \Wcal$ be a random variable. Then the estimator of $X$ based on $W$ that minimizes the expected loss can be represented as follows.
\begin{equation}
    \arg\min_{\hat{X}(\cdot):\Wcal \rightarrow \Acal_{\hat{X}}} \Ebb \left[ \mathrm{\Lambda} \left(X, \hat{X}(W) \right) \right] = \hat{X}_B \left(P_{X|W}\right)
\end{equation}

\begin{proof}
\begin{align*}
    \arg\min_{\hat{X}(\cdot):\Wcal \rightarrow \Acal_{\hat{X}}} \Ebb \left[ \mathrm{\Lambda}\left(X, \hat{X}(W)\right) \right] &= \arg\min_{\hat{X}(\cdot):\Wcal \rightarrow \Acal_{\hat{X}}} \Ebb_W \left[ \Ebb \left[ \mathrm{\Lambda} \left(X, \hat{X}(W)\right) \big\rvert W \right] \right] \\
    &= \Ebb_W \left[ \arg \min_{\hat{X}(W) \in \Acal_{\hat{X}}} \Ebb \left[ \mathrm{\Lambda} \left(X, \hat{X}(W)\right) \big\rvert W \right] \right]
\end{align*}
The minimum of the inner quantity is achieved by $\hat{X}_B\left(P_{X|W}\right)$.
\end{proof}

\begin{lemma} \label{lem:Bayes_bound}
For any vector $\mathbf{v}_1, \mathbf{v}_2 \in \mathbb{R}^{|\Acal|}$,
\[\lambda_{\hat{X}_B(\mathbf{v}_2)}^T \mathbf{v}_1 - \lambda_{\hat{X}_B(\mathbf{v}_1)}^T \mathbf{v}_1 \leq \mathrm{\Lambda}_\text{max} \left\Vert \mathbf{v}_1-\mathbf{v}_2\right\Vert_1\]
\begin{proof}
By the definition of $\hat{X}_B(\cdot)$,
\[\lambda_{\hat{X}_B(\mathbf{v}_2)}^T \mathbf{v}_2 - \lambda_{\hat{X}_B(\mathbf{v}_1)}^T \mathbf{v}_2 = \left(\lambda_{\hat{X}_B(\mathbf{v}_2)} - \lambda_{\hat{X}_B(\mathbf{v}_1)}\right)^T \mathbf{v}_2 \leq 0\]
Subtracting this nonpositive term from LHS, we can get the following inequality.
\begin{align*}
\text{LHS}&=\left(\lambda_{\hat{X}_B(\mathbf{v}_2)} - \lambda_{\hat{X}_B(\mathbf{v}_1)}\right)^T \mathbf{v}_1 \\
&\leq \left(\lambda_{\hat{X}_B(\mathbf{v}_2)} - \lambda_{\hat{X}_B(\mathbf{v}_1)}\right)^T \left( \mathbf{v}_1 - \mathbf{v}_2 \right) \\
&\leq \mathrm{\Lambda}_\text{max} \left\Vert \mathbf{v}_1 - \mathbf{v}_2\right\Vert_1,
\end{align*}
where the last inequality holds since $0\leq \lambda_i(j)\leq\mathrm{\Lambda}_\text{max}$ and $|\lambda_i(j) - \lambda_k(j)| \leq \mathrm{\Lambda}_\text{max}$ for all $i, j, k\in\Acal$.
\end{proof}
\end{lemma}

\begin{lemma} \label{lem:loss_bound}
\[\Ebb\left[\mathrm{\Lambda}\left(X, \hat{X}_B^{Q_Z}(Z)\right) - \mathrm{\Lambda}\left(X, \hat{X}_B^{P_Z}(Z)\right)\right] \leq C_1(\Pi) \mathrm{\Lambda}_\text{max} \left\Vert P_Z-Q_Z\right\Vert_1\]
where $Q_Z\in\mathbb{R}^{|\Acal|}$ is any PMF, $P_Z\in\mathbb{R}^{|\Acal|}$ is a true PMF of $Z$, and $C_1(\Pi)=\left|\Pi^{-T}\right|_\text{max} \left| \Acal_X \right|$ is a constant that depends only on $\Pi$ with $\left|\Pi^{-T}\right|_\text{max}$ denoting a maximum component of $\left|\Pi^{-T}\right|$.
\begin{proof}
    Note that, if $a>0$ is a constant,
    \[\hat{X}_B(a\mathbf{v}) = \arg \min_{\hat{x}} \lambda_{\hat{x}}^T(a\mathbf{v}) = \arg \min_{\hat{x}} \lambda_{\hat{x}}^T \mathbf{v} = \hat{X}_B(\mathbf{v}).\]
    \begin{align*}
        &\Ebb\left[\mathrm{\Lambda}\left(X, \hat{X}_B^{Q_z}(Z)\right) - \mathrm{\Lambda}\left(X, \hat{X}_B^{P_Z}(Z)\right) \bigg| Z=z \right]\\
        &\qquad =\Ebb\left[\mathrm{\Lambda}\left(X, \hat{X}_B\left(\frac{\Pi_z \odot \left(\Pi^{-T} Q_Z\right)}{Q_Z(z)}\right)\right) - \mathrm{\Lambda}\left(X, \hat{X}_B\left(\frac{\Pi_z \odot \left(\Pi^{-T} P_Z\right)}{P_Z(z)}\right)\right)\Bigg|Z=z\right]\\
        &\qquad =\Ebb\left[\mathrm{\Lambda}\left(X, \hat{X}_B\left(\Pi_z \odot \left(\Pi^{-T} Q_Z\right)\right)\right) - \mathrm{\Lambda}\left(X, \hat{X}_B\left(\Pi_z \odot \left(\Pi^{-T} P_Z\right)\right)\right)\Bigg|Z=z\right]\\
        &\qquad =\lambda_{\hat{X}_B \left(\Pi_z \odot \Pi^{-T} Q_Z\right)}^T P_{X|Z=z} - \lambda_{\hat{X}_B \left(\Pi_z \odot \Pi^{-T} P_Z\right)}^T P_{X|Z=z}
    \end{align*}
    For a constant $a\geq0$, we can get the following property by multiplying $a$ to the both sides of Lemma~\ref{lem:Bayes_bound}.
    \[\lambda_{\hat{X}_B(\mathbf{v}_2)}^T a\mathbf{v}_1 - \lambda_{\hat{X}_B(\mathbf{v}_1)}^T a\mathbf{v}_1 \leq \mathrm{\Lambda}_\text{max} a\left\Vert \mathbf{v}_1-\mathbf{v}_2\right\Vert_1\]
    By substituting $\mathbf{v}_1=\Pi_z \odot \Pi^{-T} P_Z, \mathbf{v}_2=\Pi_z \odot \Pi^{-T} Q_Z$, and $a=\frac{1}{P_Z(z)}$,
    \begin{align*}
        &\lambda_{\hat{X}_B \left(\Pi_z \odot \Pi^{-T} Q_Z\right)}^T P_{X|Z=z} - \lambda_{\hat{X}_B \left(\Pi_z \odot \Pi^{-T} P_Z\right)}^T P_{X|Z=z}\\
        &\leq \mathrm{\Lambda}_\text{max} \frac{1}{P_Z(z)} \left\Vert \Pi_z \odot \Pi^{-T} Q_Z - \Pi_z \odot \Pi^{-T} P_Z \right\Vert_1 \\
        &= \mathrm{\Lambda}_\text{max} \frac{1}{P_Z(z)} \sum_i \left|\Pi(i, z) \left( \Pi^{-T} \left( Q_Z - P_Z \right) \right)_i\right|\\
        &=\mathrm{\Lambda}_\text{max} \frac{1}{P_Z(z)} \sum_i \left|\Pi(i, z)\right| \left|\sum_j \Pi^{-T}(i,j) (Q_Z-P_Z)_j\right| \\
        &\leq \mathrm{\Lambda}_\text{max} \frac{1}{P_Z(z)} \sum_i \left|\Pi(i, z)\right| \left( \left|\Pi^{-T}\right|_\text{max} \left\Vert Q_Z-P_Z \right\Vert_1\right) \\
        &= \mathrm{\Lambda}_\text{max} \frac{1}{P_Z(z)} \left\Vert \Pi_z\right\Vert_1 \left|\Pi^{-T}\right|_\text{max} \left\Vert Q_Z-P_Z\right\Vert_1,
    \end{align*}
    where $\left|\Pi^{-T}\right|_\text{max}$ denotes a maximum component of $\left|\Pi^{-T}\right|$. By taking the expectation with respect to $Z$,
    \begin{align*}
        &\Ebb\left[\mathrm{\Lambda}\left(X, \hat{X}_B^{Q_Z}(Z)\right) - \mathrm{\Lambda}\left(X, \hat{X}_B^{P_Z}(Z)\right)\right] \\
        &=\Ebb\left[ \Ebb\left[\mathrm{\Lambda}\left(X, \hat{X}_B^{Q_z}(z)\right) - \mathrm{\Lambda}\left(X, \hat{X}_B^{P_Z}(z)\right) \bigg| Z=z \right] \right]\\
        &\leq \Ebb \left[ \mathrm{\Lambda}_\text{max} \frac{1}{P_Z(z)} \left\Vert \Pi_z\right\Vert_1 \left|\Pi^{-T}\right|_\text{max} \left\Vert Q_Z-P_Z\right\Vert_1 \right] \\
        &= \mathrm{\Lambda}_\text{max} \left|\Pi^{-T}\right|_\text{max} \left\Vert Q_Z - P_Z\right\Vert_1 \sum_z P_Z(z)  \frac{1}{P_Z(z)} \left\Vert \Pi_z\right\Vert_1  \\
        &= \mathrm{\Lambda}_\text{max} \left|\Pi^{-T}\right|_\text{max} \left| \Acal_X \right| \left\Vert Q_Z - P_Z\right\Vert_1 \qquad \left(\because \sum_z \left\Vert \Pi_z\right\Vert_1 = \sum_i \sum_j \Pi(i,j) = |\Acal_X| \right)
    \end{align*}
\end{proof}
\end{lemma}

\textbf{Proof of Theorem~\ref{thm:normalizedloss}.}
For a sequence $Z^{n}$ and corresponding estimators based on true SPA $P$ and a mismatched SPA $Q$, the expected difference in normalized loss is bounded by: 
    \[\Ebb\left[ \frac{1}{n} \sum_{t=1}^{n}\mathrm{\Lambda}\left(X_t, \hat{X}_B^{Q_{Z_t|Z^{t-1}}}\left(Z_t\right)\right) - \frac{1}{n} \sum_{t=1}^{n} \mathrm{\Lambda}\left(X_t, \hat{X}_t^\text{opt}\left(Z^t\right)\right)\right] \leq \sqrt{2} C_1(\Pi) \mathrm{\Lambda}_\text{max} \sqrt{\frac{1}{n}D \left(P_{Z^n}\Vert Q_{Z^n}\right)}, \]
where $C_1(\Pi)=\left|\Pi^{-T}\right|_\text{max} \left| \Acal_X \right|$ is a constant that depends only on $\Pi$ with $\left|\Pi^{-T}\right|_\text{max}$ denoting a maximum component of $\left|\Pi^{-T}\right|$. And $D(P_{Z^n}\Vert Q_{Z^n})$ is the relative entropy between the true and mismatched distributions. 
    
\begin{proof}
    \begin{align*}
        \text{LHS}&= \frac{1}{n} \sum_{t=1}^{n} \Ebb \left[ \Ebb \left[ \mathrm{\Lambda}\left(X_t, \hat{X}_B^{Q_{Z_t|Z^{t-1}}}(Z_t)\right) - \mathrm{\Lambda}\left(X_t, \hat{X}_B^{P_{Z_t|Z^{t-1}}}(Z_t)\right) \Big{|} Z^{t-1} \right] \right] \\
        &\leq \frac{1}{n} \sum_{t=1}^{n} \Ebb \left[ C_1(\Pi) \mathrm{\Lambda}_\text{max} \left\Vert P_{Z_t|Z^{t-1}} - Q_{Z_t|Z^{t-1}} \right\Vert_1\right] \qquad (\because \text{Lemma}~\ref{lem:loss_bound})
    \end{align*}
    By Pinsker's inequality, for any $M$-dimensional probability vectors $P$ and $Q$,
    \[\left\Vert P-Q\right\Vert_1 \leq \sqrt{2 D(P \Vert Q)}\]
    \begin{align*}
        &\frac{1}{n} \sum_{t=1}^{n} \Ebb \left[ C_1(\Pi) \mathrm{\Lambda}_\text{max} \left\Vert P_{Z_t|Z^{t-1}} - Q_{Z_t|Z^{t-1}} \right\Vert_1\right] \\
        &\leq C_1(\Pi) \mathrm{\Lambda}_\text{max} \frac{1}{n} \sum_{t=1}^{n} \Ebb \left[ \sqrt{2D\left(P_{Z_t|Z^{t-1}} \Vert Q_{Z_t|Z^{t-1}}\right)}\right] \qquad (\because \text{Pinsker's inequality}) \\
        &\leq C_1(\Pi) \mathrm{\Lambda}_\text{max} \frac{1}{n} \sum_{t=1}^{n} \sqrt{2 \Ebb \left[ D\left(P_{Z_t|Z^{t-1}} \Vert Q_{Z_t|Z^{t-1}}\right) \right]} \qquad (\because \text{Jensen's inequality}) \\
        &\leq \sqrt{2} C_1(\Pi) \mathrm{\Lambda}_\text{max} \sqrt{\frac{1}{n} \sum_{t=1}^{n} \Ebb \left[ D\left(P_{Z_t|Z^{t-1}} \Vert Q_{Z_t|Z^{t-1}}\right) \right]} \qquad (\because \text{Jensen's inequality})
    \end{align*}
    Let 
    \[D\left(P_{X|Z} \Vert Q_{X|Z} | P_Z \right) \triangleq \Ebb_{Z\sim P_Z} \left[ D\left(P_{X|Z}\Vert Q_{X|Z}\right)\right] = \sum_{z\sim\Acal_Z}D\left(P_{X|Z=z}\Vert Q_{X|Z=z}\right) P_Z(z).\]
    Then we can represent the chain rule for relative entropy as follows.
    \[D\left(P_{X,Y} \Vert Q_{X,Y}\right) = D\left(P_X \Vert Q_X\right) + D\left(P_{Y|X} \Vert Q_{Y|X} | P_X\right)\]
    By applying the chain rule, we can get
    \begin{align*}
        &\sqrt{2} C_1(\Pi) \mathrm{\Lambda}_\text{max} \sqrt{\frac{1}{n} \sum_{t=1}^{n} \Ebb \left[ D\left(P_{Z_t|Z^{t-1}} \Vert Q_{Z_t|Z^{t-1}}\right) \right]} \\
        &= \sqrt{2} C_1(\Pi) \mathrm{\Lambda}_\text{max} \sqrt{\frac{1}{n} \sum_{t=1}^{n} D\left(P_{Z_t|Z^{t-1}} \Vert Q_{Z_t|Z^{t-1}} | P_{Z^{t-1}}\right) } \\
        &= \sqrt{2} C_1(\Pi) \mathrm{\Lambda}_\text{max} \sqrt{\frac{1}{n} D\left(P_{Z^n} \Vert Q_{Z^n}\right) }
    \end{align*}
\end{proof}

\textbf{Proof of Lemma~\ref{lem:conditionaldelay}.}
For $d \geq 1$,
    \[P_{X_t|Z^{t-d}}=\Pi^{-T} P_{Z_t|Z^{t-d}}\]
\begin{proof}
    Proving the given equality is equivalent to showing $\left( \Pi^T P_{X_t|Z^{t-d}} \right) (a)= P_{Z_t|Z^{t-d}}(a)$ for all $a \in \Acal_Z$.
    \begin{align*}
        P_{Z_t|Z^{t-d}}(a) &= P\left( Z_t=a | Z^{t-d} \right) \\
        &= \sum_{a'\in \Acal_X} P\left( Z_t=a, X_t=a' | Z^{t-d} \right) \\
        &= \sum_{a'\in \Acal_X} P\left( Z_t=a | X_t=a', Z^{t-d} \right) P\left( X_t=a' | Z^{t-d} \right)\\
        &= \sum_{a'\in \Acal_X} P\left( Z_t=a | X_t=a' \right) P\left( X_t=a' | Z^{t-d} \right) \qquad \left( \because \{ Z^{t-d}, X_t, Z_t \} \text{ is a Markov chain} \right) \\
        &= \sum_{a'\in \Acal_X} \Pi(a', a) P\left( X_t=a' | Z^{t-d} \right) \\
        &= \left( \Pi^T P_{X_t|Z^{t-d}} \right) (a)
    \end{align*}
\end{proof}

\begin{lemma} \label{lem:loss_bound_2}
\[\Ebb \left[ \mathrm{\Lambda} \left(X, \hat{X}_B \left(\Pi^{-T} Q_Z\right)\right) \right] - \Ebb \left[ \mathrm{\Lambda} \left(X, \hat{X}_B \left(\Pi^{-T} P_Z\right)\right) \right] \leq \mathrm{\Lambda}_\text{max} C_1(\Pi) \left\Vert Q_Z - P_Z \right\Vert_1,\]
where $C_1(\Pi)=\left|\Pi^{-T}\right|_\text{max} \left| \Acal_X \right|$ is a constant that depends only on $\Pi$ with $\left|\Pi^{-T}\right|_\text{max}$ denoting a maximum component of $\left|\Pi^{-T}\right|$.
\begin{proof}
    \begin{align*}
        &\Ebb \left[ \mathrm{\Lambda} \left(X, \hat{X}_B \left(\Pi^{-T} Q_Z\right)\right) \right] - \Ebb \left[ \mathrm{\Lambda} \left(X, \hat{X}_B \left(\Pi^{-T} P_Z\right)\right) \right] \\
        &=\lambda_{\hat{X}_B\left(\Pi^{-T}Q_Z\right)}^T P_X - \lambda_{\hat{X}_B\left(\Pi^{-T}P_Z\right)}^T P_X.
    \end{align*}
    Note that $P_X = \Pi^{-T} P_Z$. By Lemma~\ref{lem:Bayes_bound},
    \begin{align*}
        &\lambda_{\hat{X}_B\left(\Pi^{-T}Q_Z\right)}^T P_X - \lambda_{\hat{X}_B\left(\Pi^{-T}P_Z\right)}^T P_X \\
        &\leq \mathrm{\Lambda}_\text{max} \left\Vert \Pi^{-T} P_Z - \Pi^{-T} Q_Z \right\Vert_1 \\
        &= \mathrm{\Lambda}_\text{max} \sum_i \left| \left( \Pi^{-T} (P_Z-Q_Z) \right)_i \right| \\
        &= \mathrm{\Lambda}_\text{max} \sum_i \left|\sum_j \Pi^{-T}(i,j) (P_Z-Q_Z)_j\right| \\
        &\leq \mathrm{\Lambda}_\text{max} \sum_i \left| \Pi^{-T} \right|_\text{max} \left\Vert P_Z - Q_Z \right\Vert_1 \\
        &= \mathrm{\Lambda}_\text{max} |\Acal_X| \left| \Pi^{-T} \right|_\text{max} \left\Vert P_Z - Q_Z \right\Vert_1 \\
        &= \mathrm{\Lambda}_\text{max} C_1(\Pi) \left\Vert Q_Z - P_Z \right\Vert_1
    \end{align*}
\end{proof}
\end{lemma}

\textbf{Proof of Theorem~\ref{thm:normalizedloss_delay}.} For $d \geq 1$,
    \[\Ebb\left[ \frac{1}{n} \sum_{t=1}^{n}\mathrm{\Lambda}\left(X_t, \hat{X}_B\left( \Pi^{-T}Q_{Z_t|Z^{t-d}} \right)\right) - \frac{1}{n} \sum_{t=1}^{n} \mathrm{\Lambda}\left(X_t, \hat{X}_t^\text{opt}\left(Z^{t-d}\right)\right)\right] \leq \sqrt{2} \mathrm{\Lambda}_\text{max} C_1(\Pi) \sqrt{\frac{d}{n}D_{KL}\left(P_{Z^n}\Vert Q_{Z^n}\right)}, \]
where $C_1(\Pi)=\left|\Pi^{-T}\right|_\text{max} \left| \Acal_X \right|$ is a constant that depends only on $\Pi$ with $\left|\Pi^{-T}\right|_\text{max}$ denoting a maximum component of $\left|\Pi^{-T}\right|$. 

\begin{proof}
    By using the result from Lemma~\ref{lem:loss_bound_2}, we can show that
    \begin{align*}
        \text{LHS}&=\Ebb\left[ \frac{1}{n} \sum_{t=1}^{n}\mathrm{\Lambda}\left(X_t, \hat{X}_B\left( \Pi^{-T}Q_{Z_t|Z^{t-d}} \right)\right) - \frac{1}{n} \sum_{t=1}^{n} \mathrm{\Lambda}\left(X_t, \hat{X}_B\left( \Pi^{-T}P_{Z_t|Z^{t-d}} \right)\right)\right] \\
        &= \frac{1}{n} \sum_{t=1}^{n} \Ebb \left[ \Ebb \left[ \mathrm{\Lambda}\left(X_t, \hat{X}_B\left( \Pi^{-T}Q_{Z_t|Z^{t-d}} \right)\right) - \mathrm{\Lambda}\left(X_t, \hat{X}_B\left( \Pi^{-T}P_{Z_t|Z^{t-d}} \right)\right) \Big{|} Z^{t-d} \right] \right] \\
        &\leq \frac{1}{n} \sum_{t=1}^{n} \Ebb \left[ \mathrm{\Lambda}_\text{max} C_1(\Pi) \left\Vert Q_{Z_t|Z^{t-d}} - P_{Z_t|Z^{t-d}} \right\Vert_1 \right]
    \end{align*}
    Applying Pinsker's inequality and Jensen's inequality as in Theorem~\ref{thm:normalizedloss}, we can get the following inequality.
    \begin{align*}
        &\frac{1}{n} \sum_{t=1}^{n} \Ebb \left[ \mathrm{\Lambda}_\text{max} C_1(\Pi) \left\Vert Q_{Z_t|Z^{t-d}} - P_{Z_t|Z^{t-d}} \right\Vert_1 \right] \\
        &\leq \sqrt{2} \mathrm{\Lambda}_\text{max} C_1(\Pi) \sqrt{\frac{1}{n} \sum_{t=1}^n D \left( P_{Z_t|Z^{t-d}} \Vert Q_{Z_t|Z^{t-d}} | P_{Z^{t-d}} \right)}
    \end{align*}
    Note that
    \begin{align*}
        &\sum_{t=1}^n D \left( P_{Z_t|Z^{t-d}} \Vert Q_{Z_t|Z^{t-d}} | P_{Z^{t-d}} \right) \\
        &\leq \sum_{t=1}^n D \left( P_{Z_{t-d+1}^t|Z^{t-d}} \Big{\Vert} Q_{Z_{t-d+1}^t|Z^{t-d}} \Big{|} P_{Z^{t-d}} \right) \qquad (\because \text{chain rule/data processing for relative entropy})\\
        &=\sum_{t=1}^n \sum_{k=0}^{d-1} D \left( P_{Z_{t-d+1+k}|Z_{t-d+1}^{t-d+k}, Z^{t-d}} \Big{\Vert} Q_{Z_{t-d+1+k}|Z_{t-d+1}^{t-d+k}, Z^{t-d}} \Big{|} P_{Z_{t-d+1}^{t-d+k}, Z^{t-d}} \right) \\
        &= \sum_{t=1}^n \sum_{k=0}^{d-1} D \left( P_{Z_{t-d+1+k}|Z^{t-d+k}} \Big{\Vert} Q_{Z_{t-d+1+k}|Z^{t-d+k}} \Big{|} P_{Z^{t-d+k}} \right).
    \end{align*}
    Let $\alpha_t = D\left(P_{Z_t|Z^{t-1}} \Vert Q_{Z_t|Z^{t-1}} | P_{Z^{t-1}}\right)$. Then
    \begin{align*}
        &\sum_{t=1}^n \sum_{k=0}^{d-1} D \left( P_{Z_{t-d+1+k}|Z^{t-d+k}} \Big{\Vert} Q_{Z_{t-d+1+k}|Z^{t-d+k}} \Big{|} P_{Z^{t-d+k}} \right) \\
        &= \sum_{t=1}^n \sum_{k=0}^{d-1} \alpha_{t-d+1+k} \\
        &\leq d \sum_{t=1}^{n} \alpha_t \\
        &= d D \left( P_{Z^n} \Vert Q_{Z^n} \right).
    \end{align*}
    Thus,
    \begin{align*}
        &\sqrt{2} \mathrm{\Lambda}_\text{max} C_1(\Pi) \sqrt{\frac{1}{n} \sum_{t=1}^n D \left( P_{Z_t|Z^{t-d}} \Vert Q_{Z_t|Z^{t-d}} | P_{Z^{t-d}} \right)} \\
        &\leq \sqrt{2} \mathrm{\Lambda}_\text{max} C_1(\Pi) \sqrt{\frac{d}{n} D \left( P_{Z^n} \Vert Q_{Z^n} \right) } 
    \end{align*}
\end{proof}

\textbf{Proof of Theorem~\ref{thm:normalizedloss_lookahead}.} 
\begin{align*} 
&\Ebb\left[ \frac{1}{n} \sum_{t=1}^{n}\mathrm{\Lambda}\left(X_t, \hat{X}_B^{Q_{Z_t|Z^{t-1},Z^{t+l}_{t+1}}}\left(Z_t\right)\right) - \frac{1}{n} \sum_{t=1}^{n} \mathrm{\Lambda}\left(X_t, \hat{X}_t^\text{opt} \left(Z^{t+l}\right)\right)\right] \\
&\leq \sqrt{2} C_1(\Pi) \mathrm{\Lambda}_\text{max} \sqrt{\frac{l+1}{n}D\left(P_{Z^{n+l}}\Vert Q_{Z^{n+l}}\right)},
\end{align*}
where $C_1(\Pi)=\left|\Pi^{-T}\right|_\text{max} \left| \Acal_X \right|$ is a constant that depends only on $\Pi$ with $\left|\Pi^{-T}\right|_\text{max}$ denoting a maximum component of $\left|\Pi^{-T}\right|$.

\begin{proof}
    \begin{align*}
        &\Ebb\left[ \frac{1}{n} \sum_{t=1}^{n}\mathrm{\Lambda}\left(X_t, \hat{X}_B^{Q_{Z_t|Z^{t-1}, Z_{t+1}^{t+l}}}\left(Z_t\right)\right) - \frac{1}{n} \sum_{t=1}^{n} \mathrm{\Lambda}\left(X_t, \hat{X}_t^\text{opt}\left(Z^{t+l}\right)\right)\right] \\
        &\leq \sqrt{2} C_1(\Pi) \mathrm{\Lambda}_\text{max} \sqrt{\frac{1}{n}\sum_{t=1}^{n} \Ebb \left[ D\left(P_{Z_t|Z^{t-1},Z_{t+1}^{t+l}} \Big{\Vert} Q_{Z_t|Z^{t-1},Z_{t+1}^{t+l}}\right) \right]} \\
        &=\sqrt{2} C_1(\Pi) \mathrm{\Lambda}_\text{max} \sqrt{\frac{1}{n}\sum_{t=1}^{n} D\left(P_{Z_t|Z^{t-1},Z_{t+1}^{t+l}} \Big{\Vert} Q_{Z_t|Z^{t-1},Z_{t+1}^{t+l}} \Big{|} P_{Z^{t-1}, Z_{t+1}^{t+l}}\right) }
    \end{align*}
    Thus it suffices to show
    \[ \sum_{t=1}^{n} D\left(P_{Z_t|Z^{t-1},Z_{t+1}^{t+l}} \Big{\Vert} Q_{Z_t|Z^{t-1},Z_{t+1}^{t+l}} \Big{|} P_{Z^{t-1}, Z_{t+1}^{t+l}}\right) \leq D\left(P_{Z^{n}} \big{\Vert} Q_{Z^{n}} \right) \]
    From the chain rule for relative entropy,
    \[D\left(P_{V|U} \Vert Q_{V|U} | P_U\right) \leq D\left(P_{U,V} \Vert Q_{U,V}\right).\]
    By conditioning on $Z^{t-1}$ and setting $V=Z_t, U=Z_{t+1}^{t+l}$,
    \begin{align*}
        &\sum_{t=1}^{n} D\left(P_{Z_t|Z^{t-1},Z_{t+1}^{t+l}} \Big{\Vert} Q_{Z_t|Z^{t-1},Z_{t+1}^{t+l}} \Big{|} P_{Z^{t-1}, Z_{t+1}^{t+l}}\right) \\
        &\leq \sum_{t=1}^{n} D\left(P_{Z_t^{t+l}|Z^{t-1}} \Big{\Vert} Q_{Z_t^{t+l}|Z^{t-1}} \Big{|} P_{Z^{t-1}}\right) \\
        &= \sum_{t=1}^{n} \sum_{k=0}^{l} D\left(P_{Z_{t+k}|Z^{t+k-1}} \Big{\Vert} Q_{Z_{t+k}|Z^{t+k-1}} \Big{|} P_{Z^{t+k-1}}\right) \qquad (\because \text{Chain rule})
    \end{align*}
    Let $\alpha_t = D\left(P_{Z_t|Z^{t-1}} \Vert Q_{Z_t|Z^{t-1}} | P_{Z^{t-1}}\right)$. Then
    \[ \sum_{t=1}^n \alpha_t = D\left( P_{Z^n} \Vert Q_{Z^n} \right). \]
    Using this property,
    \begin{align*}
        \sum_{t=1}^{n} D\left(P_{Z_t|Z^{t-1}} \Big{\Vert} Q_{Z_t|Z^{t-1}} \Big{|} P_{Z^{t-1}}\right) =
    \end{align*}
    \begin{align*}
        &\sum_{t=1}^{n} \sum_{k=0}^{l} D\left(P_{Z_{t+k}|Z^{t+k-1}} \Big{\Vert} Q_{Z_{t+k}|Z^{t+k-1}} \Big{|} P_{Z^{t+k-1}}\right) \\
        &=\sum_{t=1}^{n} \sum_{k=0}^{l} \alpha_{t+k} \\
        &\leq (l+1) \sum_{t=1}^{n+l} \alpha_t \\
        &=(l+1) D\left( P_{Z^{n+l}} \Vert Q_{Z^{n+l}} \right)
    \end{align*}
\end{proof}

\textbf{Proof of \ref{thm:mutual-information-lower-bounds}}.
   Suppose that $\mathcal{A}_X = \mathcal{A}_{\hat{X}} = \{0, 1, \dots, m-1\}$.
    Also, suppose the loss function is subtractive, i.e., $\mathrm{\Lambda}(x, x') = \rho(x - x')$ for some function $\rho$ with $\rho(0) = 0$, where subtraction is performed modulo $m$.

    Then, for any fixed integer $l$,
    \[\mathbb{E}\left[ \frac{1}{n} \sum_{t=1}^n \mathrm{\Lambda}\left(X_t,  \hat{X}^{opt}\left(Z^{t+l}\right)\right)\right] \geq \phi^{-1}\left( \tfrac{1}{n} H(X^n\lVert Z^{n+l}) \right),\]
     where $\phi(\cdot)$ is defined as
    \[\phi(D) \triangleq \max \left\{ H(U) : U \text{ random variable over } \mathcal{A}_X \text{ s.t. } \mathbb{E}[\rho(U)] \leq D \right\},\]
    and $H(X^n\lVert Z^{n+l})$ is the entropy of $X^n$ causally conditioned on $Z^{n+l}$,
    \[H(X^n\lVert Z^{n+l}) = \sum_{t=1}^n H(X_t|X^{t-1}, Z^{t+l}).\]
\begin{proof}
   $\phi(D)$ is concave and strictly increasing for $0 \leq D \leq \tfrac{1}{m}\sum_{i=0}^{m-1} \phi(i)$.
    As such, for any denoiser $\hat{X}_t(Z^{t+l})$ such that the average expected loss is $\leq \tfrac{1}{m}\sum_{i=0}^{m-1} \phi(i)$,\footnote{$\tfrac{1}{m}\sum_{i=0}^{m-1} \phi(i)$ is the expected loss of a trivial denoiser that chooses $\hat{X}$ uniformly over $\mathcal{A}_X$, so, even in a worst-case scenario, this condition is met for the optimal denoiser.}
    \begin{align*}
        &\phi\left(\mathbb{E} \left[\frac{1}{n}\sum_{t=1}^n \mathrm{\Lambda}\left(X_t, \hat{X}_t\left(Z^{t+l}\right)\right) \right]\right) \\\ 
        &\hspace{5em}=\phi\left(\frac{1}{n}\sum_{t=1}^n \mathbb{E}\left[\rho\left(X_t - \hat{X}_t\left(Z^{t+l}\right)\right)\right] \right) \\
        &\hspace{5em}\geq \frac{1}{n} \sum_{t=1}^n \phi\left(\mathbb{E} \left[\rho\left(X_t - \hat{X}_t\left(Z^{t+l}\right)\right)\right]\right)\qquad(\because\text{concavity of $\phi(\cdot)$; Jensen's inequality}) \\
        &\hspace{5em}= \frac{1}{n} \sum_{t=1}^n  \phi\left(\sum_{z^{t+l} \in \mathcal{A}_X^{t+l}} \Pr\left(Z^{t+l} = z^{t+l}\right) \mathbb{E}\left[\rho\left(X_t - \hat{X}_t\left(z^{t+l}\right)\right)\,\bigg\rvert Z^{t+l} = z^{t+l}\right] \right) \\
        &\hspace{5em}\geq \frac{1}{n} \sum_{t=1}^n  \sum_{z^{t+l} \in \mathcal{A}_X^{t+l}} \Pr\left(Z^{t+l} = z^{t+l}\right) \phi\left(\mathbb{E}\left[\rho\left(X_t - \hat{X}_t\left(z^{t+l}\right)\right)\,\bigg\rvert Z^{t+l} = z^{t+l}\right] \right).\qquad(\because\text{Jensen's inequality)}
    \end{align*}
    By the definition of $\phi(D)$, 
    \[ \phi\left(\mathbb{E}\left[\rho\left(X_t - \hat{X}_t\left(z^{t+l}\right)\right)\,\bigg\rvert Z^{t+l} = z^{t+l}\right] \right) \geq H\left(X_t - \hat{X}\left(z^{t+l}\right)\big|Z^{t+l}=z^{t+l}\right) = H\left(X_t\big|Z^{t+l}=z^{t+l}\right).\]
    As a result,
    \begin{align*}
         \phi\left(\mathbb{E} \left[\frac{1}{n}\sum_{t=1}^n \mathrm{\Lambda}\left(X_t, \hat{X}_t\left(Z^{t+l}\right)\right) \right]\right)
        &\geq \frac{1}{n}\sum_{t=1}^n \sum_{z^{t+l} \in \mathcal{A}_X^{t+l}} \Pr\left(Z^{t+l} = z^{t+l}\right) H\left(X_t\big|Z^{t+l}=z^{t+l}\right) \\
        &=  \frac{1}{n}\sum_{t=1}^n H\left(X_t\big|Z^{t+l}\right) \\
        &\geq \frac{1}{n}\sum_{t=1}^n H\left(X_t\big|Z^{t+l}, X^{t-1}\right)\qquad(\because \text{conditioning reduces entropy}) \\
        &= \tfrac{1}{n} H(X^n \lVert Z^{n+l}).
    \end{align*}
    By the strict monotonicity of $\phi(\cdot)$, $\phi^{-1}$ is well-defined over $[0, \log(M)]$ and, if $\phi(D) \geq \alpha$, then $D \geq \phi^{-1}(\alpha)$.
    $\therefore$,
    \[\mathbb{E} \left[\frac{1}{n}\sum_{t=1}^n \mathrm{\Lambda}\left(X_t, \hat{X}^{opt}_t\left(Z^{t+l}\right)\right) \right] \geq \phi^{-1}\left(\tfrac{1}{n} H(X^n||Z^{n+l})\right).\]
\end{proof}

\begin{lemma} \label{lem:loss_bound_arbitrary}
\[\left|\Ebb \left[ \mathrm{\Lambda} \left(X, \hat{X}_B \left(\Pi^{-T} Q_Z^1\right)\right) \right] - \Ebb \left[ \mathrm{\Lambda} \left(X, \hat{X}_B \left(\Pi^{-T} Q_Z^2\right)\right) \right]\right| \leq \mathrm{\Lambda}_\text{max} C_{1}(\Pi) \left( \left\Vert Q_Z^1 - P_Z \right\Vert_1 + \left\Vert Q_Z^2 - P_Z \right\Vert_1 \right) \]
where $Q_Z^1$ and $Q_Z^2$ are arbitrary PMFs, $P_Z$ is a PMF of $Z$, and $C_1(\Pi)=\left|\Pi^{-T}\right|_\text{max} \left| \Acal_X \right|$ is a constant that depends only on $\Pi$ with $\left|\Pi^{-T}\right|_\text{max}$ denoting a maximum component of $\left|\Pi^{-T}\right|$.
\begin{proof}
    \begin{align*}
        &\left|\Ebb \left[ \mathrm{\Lambda} \left(X, \hat{X}_B \left(\Pi^{-T} Q_Z^1\right)\right) \right] - \Ebb \left[ \mathrm{\Lambda} \left(X, \hat{X}_B \left(\Pi^{-T} Q_Z^2\right)\right) \right]\right| \\
        =& \left| \Ebb \left[ \mathrm{\Lambda} \left(X, \hat{X}_B \left(\Pi^{-T} Q_Z^1\right)\right) \right] - \Ebb \left[ \mathrm{\Lambda} \left(X, \hat{X}_B \left(\Pi^{-T} P_Z\right)\right) \right] +\Ebb \left[ \mathrm{\Lambda} \left(X, \hat{X}_B \left(\Pi^{-T} P_Z\right)\right) \right] - \Ebb \left[ \mathrm{\Lambda} \left(X, \hat{X}_B \left(\Pi^{-T} Q_Z^2\right)\right) \right] \right| \\
        \leq& \left| \Ebb \left[ \mathrm{\Lambda} \left(X, \hat{X}_B \left(\Pi^{-T} Q_Z^1\right)\right) \right] - \Ebb \left[ \mathrm{\Lambda} \left(X, \hat{X}_B \left(\Pi^{-T} P_Z\right)\right) \right] \right| + \left| \Ebb \left[ \mathrm{\Lambda} \left(X, \hat{X}_B \left(\Pi^{-T} P_Z\right)\right) \right] - \Ebb \left[ \mathrm{\Lambda} \left(X, \hat{X}_B \left(\Pi^{-T} Q_Z^2\right)\right) \right] \right| \\
        =& \left( \Ebb \left[ \mathrm{\Lambda} \left(X, \hat{X}_B \left(\Pi^{-T} Q_Z^1\right)\right) \right] - \Ebb \left[ \mathrm{\Lambda} \left(X, \hat{X}_B \left(\Pi^{-T} P_Z\right)\right) \right] \right) \\
        &+ \left( \Ebb \left[ \mathrm{\Lambda} \left(X, \hat{X}_B \left(\Pi^{-T} Q_Z^2\right)\right) \right] - \Ebb \left[ \mathrm{\Lambda} \left(X, \hat{X}_B \left(\Pi^{-T} P_Z\right)\right) \right] \right) \qquad (\because \Pi^{-T} P_Z=P_X, \text{definition of } \hat{X}_B) \\
        \leq& \mathrm{\Lambda}_\text{max} |\Acal_X| \left| \Pi^{-T} \right|_\text{max} \left( \left\Vert P_Z - Q_Z^1 \right\Vert_1 + \left\Vert P_Z - Q_Z^2 \right\Vert_1\right)
    \end{align*}
\end{proof}
\end{lemma}

\begin{lemma} \label{lem:empirical_distr_bound}
Let $Q$ be a PMF over $\Acal=\{ 1, 2, ..., A \}$ and $S_i, i\in \{1, 2, ..., n\}$ be iid random variables having probability distribution $Q$. Let $\delta_{S_i} \in \mathbb{R}^A$ be a one-hot vector which has element 1 on $S_i$-th index. Then, 
\[\Ebb\left[ \left\Vert Q - \frac{1}{n} \sum_{i=1}^n \delta_{S_i} \right\Vert_1 \right] \leq 2^{|\Acal|} \sqrt{\frac{\pi}{2}} \frac{1}{\sqrt{n}}\] 
\begin{proof}
    By Theorem 2.1 in \cite{verdu2003inequalities},
    \[\Pr \left( \left\Vert Q - \frac{1}{n} \sum_{i=1}^n \delta_{S_i} \right\Vert_1 \geq \epsilon \right) \leq 2^{|\Acal|} e^{-n\frac{\epsilon^2}{2}}, \forall \epsilon > 0\]
    Note that the exponent in the bound is originally $-n\varphi(\pi_Q)\epsilon^2/4$. For simplicity, we used the fact that $\varphi(\pi_Q)$ is lower bounded by 2. To derive a tighter bound, we can use the original form containing $\varphi(\pi_Q)$. \\
    By using the tail integral formula for expectation,
    \begin{align*}
        \Ebb \left[ \left\Vert Q - \frac{1}{n} \sum_{i=1}^n \delta_{S_i} \right\Vert_1 \right] &= \int_0^\infty \Pr \left( \left\Vert Q - \frac{1}{n} \sum_{i=1}^n \delta_{S_i} \right\Vert_1 \geq \epsilon \right) d\epsilon\\
        &\leq 2^{|\Acal|} \int_0^\infty e^{-n\frac{\epsilon^2}{2}} d\epsilon \\
        &=2^{|\Acal|} \sqrt{\frac{\pi}{2}} \frac{1}{\sqrt{n}}
    \end{align*}
\end{proof}
\end{lemma}

\textbf{Proof of Theorem~\ref{thm:montecarlo}.} 
    Let $\hat{Q}_{Z_t|Z^{t-d}}^M$ be an empirical distribution calculated by sampling the Monte Carlo realization $M$ times. 
\begin{align*}
    &\Ebb\left[ \frac{1}{n} \sum_{t=1}^{n}\mathrm{\Lambda}\left(X_t, \hat{X}_B\left(\Pi^{-T}\hat{Q}_{Z_t|Z^{t-d}}^M \right)\right) - \frac{1}{n} \sum_{t=1}^{n} \mathrm{\Lambda}\left(X_t, \hat{X}_t^{\text{opt}}\left(Z^{t-d}\right)\right)\right] \\
    &\leq \mathrm{\Lambda}_\text{max} \left( C_{2}(\Pi) \sqrt{\frac{d}{n}D\left(P_{Z^n}\Vert Q_{Z^n}\right)} + C_{3}(\Pi) \frac{1}{\sqrt{M}} \right),
\end{align*}
where $C_2(\Pi)$ and $C_3(\Pi)$ are constants that depend only on $\Pi$, defined as follows.
\[C_2(\Pi) = 3 \sqrt{2} \left|\Pi^{-T}\right|_\text{max} \left| \Acal_X \right| \]
\[C_3(\Pi) = \left|\Pi^{-T}\right|_\text{max} \left| \Acal_X \right| 2^{|\Acal_Z|} \sqrt{\frac{\pi}{2}} \]

\begin{proof}
    \begin{align*}
        LHS =&\,\Ebb\left[ \frac{1}{n} \sum_{t=1}^{n}\mathrm{\Lambda}\left(X_t, \hat{X}_B\left(\Pi^{-T}\hat{Q}_{Z_t|Z^{t-d}}^M \right)\right) - \frac{1}{n} \sum_{t=1}^{n} \mathrm{\Lambda}\left(X_t, \hat{X}_B\left(\Pi^{-T}P_{Z_t|Z^{t-d}} \right)\right)\right] \\
        =&\, \Ebb \left[ \frac{1}{n} \sum_{t=1}^{n}\mathrm{\Lambda}\left(X_t, \hat{X}_B\left(\Pi^{-T}\hat{Q}_{Z_t|Z^{t-d}}^M \right)\right) - \frac{1}{n} \sum_{t=1}^{n}\mathrm{\Lambda}\left(X_t, \hat{X}_B\left(\Pi^{-T}Q_{Z_t|Z^{t-d}} \right)\right) \right] \\
        &+ \Ebb \left[ \frac{1}{n} \sum_{t=1}^{n}\mathrm{\Lambda}\left(X_t, \hat{X}_B\left(\Pi^{-T}Q_{Z_t|Z^{t-d}} \right)\right) - \frac{1}{n} \sum_{t=1}^{n}\mathrm{\Lambda}\left(X_t, \hat{X}_B\left(\Pi^{-T}P_{Z_t|Z^{t-d}} \right)\right) \right] \\
    \end{align*}
    We can get the bound of the second term by Theorem~\ref{thm:normalizedloss_delay} as follows.
    \[\Ebb \left[ \frac{1}{n} \sum_{t=1}^{n}\mathrm{\Lambda}\left(X_t, \hat{X}_B\left(\Pi^{-T}Q_{Z_t|Z^{t-d}} \right)\right) - \frac{1}{n} \sum_{t=1}^{n}\mathrm{\Lambda}\left(X_t, \hat{X}_B\left(\Pi^{-T}P_{Z_t|Z^{t-d}} \right)\right) \right] \leq \sqrt{2} \mathrm{\Lambda}_\text{max} C_1(\Pi) \sqrt{\frac{d}{n}D\left(P_{Z^n}\Vert Q_{Z^n}\right)}\]
    For the first term,
    \begin{align*}
        &\Ebb \left[ \frac{1}{n} \sum_{t=1}^{n}\mathrm{\Lambda}\left(X_t, \hat{X}_B\left(\Pi^{-T}\hat{Q}_{Z_t|Z^{t-d}}^M \right)\right) - \frac{1}{n} \sum_{t=1}^{n}\mathrm{\Lambda}\left(X_t, \hat{X}_B\left(\Pi^{-T}Q_{Z_t|Z^{t-d}} \right)\right) \right] \\
        &=\frac{1}{n} \sum_{t=1}^{n} \Ebb \left[ \Ebb \left[ \Ebb \left[ \mathrm{\Lambda}\left(X_t, \hat{X}_B\left(\Pi^{-T}\hat{Q}_{Z_t|Z^{t-d}}^M \right)\right) - \mathrm{\Lambda}\left(X_t, \hat{X}_B\left(\Pi^{-T}Q_{Z_t|Z^{t-d}} \right)\right) \Big{|} Z^{t-d}, M \text{realizations} \right] \Big{|} Z^{t-d} \right] \right] \\
        &\leq \frac{1}{n} \sum_{t=1}^{n} \Ebb \left[ \Ebb \left[ \mathrm{\Lambda}_\text{max} C_1(\Pi) \left( \left\Vert \hat{Q}_{Z_t|Z^{t-d}}^M - P_{Z_t|Z^{t-d}} \right\Vert_1 + \left\Vert Q_{Z_t|Z^{t-d}} - P_{Z_t|Z^{t-d}} \right\Vert_1 \right) \Big{|} Z^{t-d} \right] \right] \qquad (\because \text{Lemma~\ref{lem:loss_bound_arbitrary}} )\\
        &\leq \frac{1}{n} \mathrm{\Lambda}_\text{max} C_1(\Pi) \sum_{t=1}^{n} \Ebb \left[ \Ebb \left[ \left\Vert \hat{Q}_{Z_t|Z^{t-d}}^M - Q_{Z_t|Z^{t-d}} \right\Vert_1 + 2 \left\Vert Q_{Z_t|Z^{t-d}} - P_{Z_t|Z^{t-d}} \right\Vert_1 \Big{|} Z^{t-d} \right] \right] \qquad (\because \text{triangle inequality}) \\
        &\leq \frac{1}{n} \mathrm{\Lambda}_\text{max} C_1(\Pi) \sum_{t=1}^{n} \Ebb \left[ \Ebb \left[ \left\Vert \hat{Q}_{Z_t|Z^{t-d}}^M - Q_{Z_t|Z^{t-d}} \right\Vert_1 \Big{|} Z^{t-d} \right] \right] + 2 \sqrt{2} \mathrm{\Lambda}_\text{max} C_1(\Pi) \sqrt{\frac{d}{n}D\left(P_{Z^n}\Vert Q_{Z^n}\right)} \qquad (\because \text{Theorem~\ref{thm:normalizedloss_delay}})\\
        &\leq \frac{1}{n} \mathrm{\Lambda}_\text{max} C_1(\Pi) \sum_{t=1}^{n} \Ebb \left[ 2^{|\Acal_Z|} \sqrt{\frac{\pi}{2}} \frac{1}{\sqrt{M}} \right] + 2 \sqrt{2} \mathrm{\Lambda}_\text{max} C_1(\Pi) \sqrt{\frac{d}{n}D\left(P_{Z^n}\Vert Q_{Z^n}\right)} \qquad (\because \text{Lemma~\ref{lem:empirical_distr_bound}})\\
        &= \mathrm{\Lambda}_\text{max} C_1(\Pi) 2^{|\Acal_Z|} \sqrt{\frac{\pi}{2}} \frac{1}{\sqrt{M}} + 2 \sqrt{2} \mathrm{\Lambda}_\text{max} C_1(\Pi) \sqrt{\frac{d}{n}D\left(P_{Z^n}\Vert Q_{Z^n}\right)}
    \end{align*}
    Thus,
    \begin{align*}
    &\Ebb\left[ \frac{1}{n} \sum_{t=1}^{n}\mathrm{\Lambda}\left(X_t, \hat{X}_B\left(\Pi^{-T}\hat{Q}_{Z_t|Z^{t-d}}^M \right)\right) - \frac{1}{n} \sum_{t=1}^{n} \mathrm{\Lambda}\left(X_t, \hat{X}_B\left(\Pi^{-T}P_{Z_t|Z^{t-d}} \right)\right)\right] \\
    &\leq \mathrm{\Lambda}_\text{max} C_1(\Pi) 2^{|\Acal_Z|} \sqrt{\frac{\pi}{2}} \frac{1}{\sqrt{M}} + 3 \sqrt{2} \mathrm{\Lambda}_\text{max} C_1(\Pi) \sqrt{\frac{d}{n}D\left(P_{Z^n}\Vert Q_{Z^n}\right)}
    \end{align*}
\end{proof}

\bibliographystyle{IEEEtran}
\bibliography{main.bib}

\end{document}